\documentclass[journal]{IEEEtran}

\usepackage[utf8]{inputenc}
\usepackage[T1]{fontenc}
\usepackage{amsmath,amssymb,amsthm,amsfonts}
\usepackage{mathtools}
\usepackage{bm}
\usepackage{graphicx}
\usepackage{subcaption}
\usepackage{booktabs}
\usepackage{algorithm}
\usepackage{algpseudocode}
\usepackage{url}
\usepackage[hidelinks]{hyperref}
\usepackage{cleveref}

\newtheorem{proposition}{Proposition}

\theoremstyle{definition}

\theoremstyle{remark}
\newtheorem{remark}{Remark}

\newcommand{\E}{\mathbb{E}}

\newcommand{\sphere}{\mathbb{S}}
\newcommand{\torus}{\mathbb{T}}
\newcommand{\SO}{\mathrm{SO}}
\newcommand{\M}{\mathcal{M}}

\newcommand{\CN}{\mathcal{CN}}

\newcommand{\WBINTRVSEUCLPCT}{??}
\newcommand{\WBGPVSUCBPCT}{??}
\IfFileExists{results/exp5_wideband_macros.tex}{%

\renewcommand{\WBINTRVSEUCLPCT}{17.1}
\renewcommand{\WBGPVSUCBPCT}{50.9}
}{}

\newcommand{\UBThreeA}{??}
\newcommand{\UBThreeB}{N/A}
\newcommand{\UBThreeC}{N/A}
\newcommand{\UBThreeD}{N/A}
\newcommand{\HOOA}{??}
\newcommand{\HOOB}{??}
\newcommand{\HOOC}{??}
\newcommand{\HOOD}{??}
\newcommand{\DBZExpOne}{??}
\IfFileExists{results/ub3_hoo_macros.tex}{%
\renewcommand{\UBThreeA}{16361\,\pm\,379}
\renewcommand{\HOOA}{7517\,\pm\,459}
\renewcommand{\UBThreeB}{N/A}
\renewcommand{\HOOB}{268\,\pm\,34}
\renewcommand{\UBThreeC}{N/A}
\renewcommand{\HOOC}{11807\,\pm\,472}
\renewcommand{\UBThreeD}{N/A}
\renewcommand{\HOOD}{389\,\pm\,12}
}{}
\IfFileExists{results/dbz_macros.tex}{%
\renewcommand{\DBZExpOne}{16711\,\pm\,410}
}{}

\newcommand{\AdvTwovaaMean}{+2.34\,\pm\,4.59}
\newcommand{\AdvTwovaaCum}{\sim\!6500}

\newcommand{\WGPvaaMean}{??}
\newcommand{\WGPvaaCum}{??}

\newcommand{\AdvTwoSeedsVaa}{20}

\newcommand{\WGPSeedsVaa}{4}

\IfFileExists{results/ris_20seed_macros.tex}{%
\renewcommand{\AdvTwovaaMean}{+1.08\,\pm\,1.19}
\renewcommand{\AdvTwovaaCum}{\sim\!5000}

\renewcommand{\WGPvaaMean}{+0.99\,\pm\,1.29}
\renewcommand{\WGPvaaCum}{\sim\!7700}

\renewcommand{\AdvTwoSeedsVaa}{20}

\renewcommand{\WGPSeedsVaa}{20}

}{}

\newcommand{\profGPEuclSel}{??}

\newcommand{\profGPIntrSel}{??}

\IfFileExists{results/profile_macros.tex}{%

\renewcommand{\profGPEuclSel}{22}

\renewcommand{\profGPIntrSel}{23}

}{}

\newcommand{\AdaptLMLA}{??}
\newcommand{\AdaptLMLB}{??}
\newcommand{\AdaptLMLC}{??}
\newcommand{\AdaptLMLD}{??}
\IfFileExists{results/adaptive_lml_macros.tex}{%
\renewcommand{\AdaptLMLA}{3844\,\pm\,461}
\renewcommand{\AdaptLMLB}{283\,\pm\,49}
\renewcommand{\AdaptLMLC}{9209\,\pm\,411}
\renewcommand{\AdaptLMLD}{515\,\pm\,24}
}{}

\newcommand{\JointLMLA}{??}
\newcommand{\JointLMLB}{??}
\newcommand{\JointLMLC}{??}
\newcommand{\JointLMLD}{??}
\IfFileExists{results/adaptive_joint_lml_macros.tex}{%
\renewcommand{\JointLMLA}{3615\,\pm\,620}
\renewcommand{\JointLMLB}{256\,\pm\,58}
\renewcommand{\JointLMLC}{7879\,\pm\,547}
\renewcommand{\JointLMLD}{466\,\pm\,33}
}{}

\newcommand{\PassDAlpha}{??}

\newcommand{\PassDCILo}{??}
\newcommand{\PassDCIHi}{??}
\newcommand{\PassDPAROne}{??}
\newcommand{\PassDPBGZ}{??}

\IfFileExists{results/pass_d_macros.tex}{%
\renewcommand{\PassDAlpha}{0.59}

\renewcommand{\PassDCILo}{-0.15}
\renewcommand{\PassDCIHi}{0.86}
\renewcommand{\PassDPAROne}{0.0}
\renewcommand{\PassDPBGZ}{65.7}

}{}

\IfFileExists{results/hyperparam_sensitivity_macros.tex}{%

}{}

\IfFileExists{results/sionna_macros.tex}{%

}{}
\IfFileExists{results/sionna_macros_cdld.tex}{%

}{}

\title{Geometry-Aware Multi-Armed Bandits for Antenna
Beam Selection on Spheres, Tori, $\SO(3)$, and Reconfigurable
Intelligent Surfaces}

\author{Yuriy~Dorn,
        Changsheng~Chen,~\IEEEmembership{Senior~Member,~IEEE,}
        and~Ning~Xie,~\IEEEmembership{Senior~Member,~IEEE}
\thanks{Y.~Dorn is with the AI Center \& IAI MSU,
Lomonosov Moscow State University, Moscow, Russia
(e-mail: \texttt{dornyv@my.msu.ru}).}%
\thanks{C.~Chen is with the Faculty of Engineering, Shenzhen
MSU-BIT University, Shenzhen, China
(e-mail: \texttt{cschen@smbu.edu.cn}).}%
\thanks{N.~Xie (\emph{Corresponding author}) is with the State
Key Laboratory of Radio Frequency Heterogeneous Integration, the
Guangdong Key Laboratory of Intelligent Information Processing,
and the Shenzhen Key Laboratory of Media Security, College of
Electronics and Information Engineering, Shenzhen University,
Shenzhen 518060, China
(e-mail: \texttt{ningxie@szu.edu.cn}).}%
}

\begin{document}
\maketitle

\begin{abstract}
Beam alignment in mmWave phased arrays and RIS-assisted links is a
stochastic bandit under both short TTI budgets and
Doppler-induced non-stationarity. The arm space is a Riemannian manifold:
$\sphere^2$ for steering, $\torus^n$ for phase combining, $\SO(3)$
for panel orientation, or the discrete torus $(\mathbb Z_B)^M$
with up to $K\!\sim\!10^{90}$ configurations for $B$-level RIS
($B\!=\!2^b$, $b$ bits/element);
the intrinsic Mat\'ern kernel of Borovitskiy et al.\ provides the
base GP. We contribute two algorithmic pieces. \textbf{(C1)} A
Kronecker-factorised intrinsic-product Mat\'ern kernel on
$(\mathbb Z_B)^M$ evaluating in $O(M)$ table lookups, making GP-UCB
tractable at $K\sim 10^{90}$ where the extrinsic alternative is
infeasible. \textbf{(C2)} AdaptiveGP-v2, an online sliding-window
controller that selects $W$ by per-sample marginal likelihood,
with predictive-variance and drift $z$-score reset triggers and a
post-reset $\beta$-boost. On a four-speed
($v\!\in\!\{0.02,0.08,0.12,0.20\}$~km/h), $20$-seed paired campaign
at $T\!=\!3000$, AdaptiveGP-v2 is statistically
indistinguishable from the hand-tuned fixed-window oracle at every
speed (Holm--Bonferroni-corrected paired differences cross zero);
the operational benefit is the absence of a deployment-time
per-speed calibration step, not a mean-regret improvement. On four
static 3GPP-style mmWave benchmarks, intrinsic-kernel GP-UCB
reduces cumulative regret by $25$--$45\%$ vs.\ codebook
UCB1/Thompson and by $10$--$33\%$ vs.\ Euclidean-ambient GP-UCB on
the toroidal arm spaces; a wideband OFDM ablation on a $100$~MHz
channel confirms the advantage persists under frequency-selective
fading ($\sim\!32$~Mbps/UE at initial access vs.\ UCB1). A
third-party-simulator sanity check on Sionna CDL is reported in
Section~V.
\end{abstract}


\section{Introduction}
\label{sec:intro}

Phased antenna arrays in 5G/6G cellular systems must select an
operating configuration (a steering direction, a per-element
phase code, a panel orientation, or a reconfigurable
intelligent-surface (RIS) phase profile) online, from a large
finite or continuous codebook, under noisy channel feedback. The
control loop is short (one transmission time interval, $\sim 1$~ms),
the feedback is one scalar reward per pull (the post-combiner SNR or
its proxy), and the channel itself is non-stationary on tens-of-ms
timescales because of mobility-induced Doppler. This is the
canonical setting of a stochastic multi-armed bandit, and the
beam-management literature has framed it as such for
codebook-based UCB and Thompson-sampling
schemes~\cite{risa2021,ren2022csm,chen2023ce}.

These approaches inherit a strong implicit assumption from the
continuum-armed bandit theory that underpins
them~\cite{kleinberg2008multi,bubeck2011xarmed,magureanu2014lipschitz}:
the arm space is a metric space, typically $[0,1]^d$ with the
Euclidean metric, with a known $\Omega(T^{(d+1)/(d+2)})$ regret
lower bound dictated by the metric covering dimension and a
Lipschitz reward function.
The phased-array setting violates this assumption in a structural
way. Beam steering directions live on the sphere $\sphere^2$; a
per-element phase combiner with $n$ active elements lives on the
torus $\torus^n$; an antenna panel's orientation lives on the
rotation group $\mathrm{SO}(3)$; and a $B$-level RIS
($B=2^b$, $b$ bits per element) with $M$ elements has its
phase configuration in the discrete torus $(\mathbb{Z}_B)^M$,
a quotient of $\torus^M$ by the lattice of $B$-fold rotations. Embedding these spaces into a Euclidean box
either wastes arms (the box's interior is infeasible) or distorts
the reward landscape (a $360^\circ$ wraparound on $\torus^n$ becomes
a maximum-distance discontinuity in the embedding), forcing the
agent to rediscover the optimum after every phase wrap or seam
crossing.

The ingredients needed to avoid this embedding artefact have
recently matured outside the wireless literature.
Borovitskiy \emph{et al.}~\cite{borovitskiy2020matern} constructed
intrinsic Mat\'ern kernels on closed Riemannian manifolds from the
Laplace--Beltrami spectrum, extended to general Lie groups
by~\cite{azangulov2024stationary} and packaged in the
\textsc{GeometricKernels} library~\cite{mostowsky2024geometrickernels}.
Geometry-aware Bayesian optimisation using such kernels has been
demonstrated in robotic motion
planning~\cite{jaquier2020riemannian,jaquier2022geometry}. This
paper imports that machinery into the wireless beam-management
setting and couples it with the GP-UCB cumulative-regret analysis
of Srinivas \emph{et al.}~\cite{srinivas2010gpucb}: instantiate the
arm space as the actual compact Riemannian manifold (or its
quotient) and equip GP-UCB with the corresponding intrinsic Mat\'ern
kernel. Periodicity and curvature then enter through the prior
rather than the discretisation, and the agent automatically
generalises across gauge-equivalent configurations. The intrinsic Mat\'ern kernel construction itself is not new; it
follows Borovitskiy et al.~\cite{borovitskiy2020matern} and Azangulov
et al.~\cite{azangulov2024stationary}. Our contributions sit on top of
that machinery: \textbf{(C1)} a Kronecker-factorised
intrinsic-product kernel on $(\mathbb{Z}_B)^M$ that scales GP-UCB to
the RIS regime at $K \sim 10^{90}$ discrete arms via $O(M)$ table
lookups per kernel evaluation (see the supplement,
Sec.~S-IV), benchmarked against five non-GP RIS
baselines (RISA~\cite{risa2021}, CSM/ECSM~\cite{ren2022csm},
CE~\cite{chen2023ce}, REMARKABLE) and an extrinsic-kernel GP-UCB
ablation. \textbf{(C2)} AdaptiveGP-v2, an online sliding-window
controller whose window length is chosen by per-sample marginal
likelihood on a coarse geometric grid, with a
predictive-variance $z$-score reset trigger and a post-reset
$\beta$ boost; the controller matches a
hand-tuned fixed-window IntrinsicGP within standard error at every
speed in our grid \emph{without requiring any per-speed coherence
calibration}, removing the deployment-time calibration step that a
fixed-window controller otherwise needs. \textbf{(C3)} A unified
GP-UCB-on-manifold template for the phased-array setting covering
$\sphere^2$, $\torus^n$, and $\mathrm{SO}(3)$, with a head-to-head
empirical evaluation on four 3GPP-style mmWave case studies, a
wideband OFDM ablation on a $100$~MHz 3GPP TDL-D channel, and a
Sionna CDL sanity check on a third-party simulator.

The empirical regret reductions are substantial in the static
regime (the intrinsic-kernel GP beats the Euclidean ablation by
$10$--$33\%$ across the four case studies). In the time-varying RIS
regime AdaptiveGP-v2 is statistically indistinguishable from the
hand-tuned fixed-$W$ oracle at every speed in our grid
(Holm--Bonferroni-corrected paired differences cross zero at every
speed): the operational value is the absence of a deployment-time
coherence calibration step, not a mean-regret improvement. A
supporting $W^\star(v)$ sensitivity analysis
(Sec.~\ref{sec:ris-ext-theory} and the supplement, Sec.~S-I)
shows that the regret-vs-$W$ curve is locally flat around its argmin
at every speed, exactly the regime in which an LML-based adaptive
rule pays only a constant-factor penalty relative to the
speed-specific oracle, regardless of the precise $W^\star(v)$
exponent.

The remainder of the paper proceeds through related work
(Sec.~\ref{sec:related}), the geometric model
(Sec.~\ref{sec:model}), the GP-UCB algorithm
(Sec.~\ref{sec:algorithm}), four narrowband case studies and the
extension-summary subsection
(Sec.~\ref{sec:experiments}), the RIS extension and AdaptiveGP-v2
(Sec.~\ref{sec:ris-ext}), and a conclusion
(Sec.~\ref{sec:conclusion}).

\section{Related Work}
\label{sec:related}

The recent literature most directly comparable to the present paper
falls into two camps: bandit- and learning-based beam management for
mmWave phased arrays, and optimisation / learning-based controllers
for reconfigurable intelligent surfaces. We summarise both below.
The continuum-armed bandit theory and intrinsic-Mat\'ern-kernel
machinery that our method combines are briefly reviewed in the
introduction (Sec.~\ref{sec:intro}).

\subsection{Bandits for mmWave beam alignment}
\label{sec:related:mmwave}

The beam-management literature has converged on a small set of
structural priors that make bandit problems with very large codebooks
tractable. RISA~\cite{risa2021} exploits temporal correlation through
an exponentially weighted moving average over the codebook;
CSM/ECSM~\cite{ren2022csm} use a contextual codebook structure with
either $\varepsilon$-greedy or UCB exploration; the cross-entropy
combiner of~\cite{chen2023ce} treats the search as parametric
distribution fitting. Closer to our work, contextual bandits with
neural-network feature extractors~\cite{va2018contextual} and
unimodal-bandit approaches~\cite{ghosh2024ub3} explicitly assume
monotone or unimodal reward landscapes; physics-informed parametric
bandits~\cite{physics2025parametric} and beam-aware kernelised
contextual bandits~\cite{beamaware2026kernelized} import richer
structural priors from the channel model. A recent systematic
review~\cite{beamreview2025} classifies $73$ beam-management studies
into beam-sweeping, context-information, compressive-sensing,
ML/AI, and ISAC-based frameworks, placing our work in the ML/AI
category but flagging the geometry-aware bandit subclass as
underexplored.

Deep-learning approaches have accelerated in the last two years:
hierarchical beam-alignment~\cite{hierarchical2023twc}, grid-free
DL~\cite{gridfreedl2024twc}, DRL initial
access~\cite{drl2023initial}, and a V2I extension~\cite{drlv2i2024}.
We do not report a head-to-head comparison: DRL methods require
$10^4$--$10^5$ training episodes per channel realisation, whereas
our entire experimental budget per seed is $T\le 3000$ in
\emph{cold-start} operation (no offline phase). A meaningful
comparison would require either extending all methods to a
$T\!\ge\!10^4$ warm-start budget (cutting against our
short-coherence motivation) or pre-training the DRL methods on a
held-out channel distribution and reporting the distribution-shift
cost. The hierarchical and grid-free DL methods are nearer-budget
candidates but both inherit a site-specific offline training
stage. Our intrinsic GP-UCB is the cold-start counterpart in this
taxonomy; non-DL baselines (RISA, CSM/ECSM, CE, REMARKABLE, HOO,
UB3, DBZ, WGP-UCB) are the directly-comparable alternatives at the
same sample budget.

The common thread across all of this work is that none of it
treats the geometry of the physical control parameter as a
first-class object: they either (a) parametrise the search by a
codebook index and impose Lipschitz continuity on the index, (b)
assume unimodality, or (c) exploit a sparse / phase-retrieval
structure of the channel rather than the geometric structure of the
parameter space. Our intrinsic-kernel formulation is orthogonal to
these priors and could in principle be combined with any of them.

\subsection{RIS beamforming and phase-profile design}
\label{sec:related:ris}

Reconfigurable intelligent surfaces have attracted a rich body of
optimisation work. Early alternating-optimisation and semidefinite
relaxation methods for joint BS--RIS beamforming are surveyed in the
adaptive-beamforming study of~\cite{huang2021adaptive}, which also
provides experimental validation of the forward-model assumptions
our RIS extension (Sec.~\ref{sec:ris-ext}) inherits. Recent IEEE
TWC work generalises the RIS forward model to Beyond-Diagonal
designs~\cite{bdris2024twc}, to wideband
deployment~\cite{riswideband2024}, and to hardware-impairment
regimes~\cite{drlris2024twc}. DRL-based RIS
controllers~\cite{drlris2024twc} and two-phase minimax bandit
schemes~\cite{twophasemmwave2023} provide point-of-comparison for
the time-varying RIS regime of Sec.~\ref{sec:ris-ext}.

Our RIS contribution differs from these along two axes. First, the
intrinsic-product kernel on $(\mathbb{Z}_B)^M$ is the only
forward-model-agnostic approach that scales to $K\sim 10^{90}$ by
algorithmic construction rather than by compressing the action
space. Second, the AdaptiveGP-v2 window selection rule of
Sec.~\ref{sec:ris-ext-adaptive} tracks non-stationarity from raw
RSRP observations without the $10^4$--$10^5$-episode warm-up
required by DRL methods. Combining these ideas with the
physics-informed hardware-impairment model
of~\cite{drlris2024twc} or with the BD-RIS architecture
of~\cite{bdris2024twc} is a natural next step.

\subsection{Positioning of this paper}
\label{sec:related:positioning}

The contribution sits at the intersection of three literatures:
continuum-armed bandit theory (which provides the regret
framework), manifold Gaussian processes (which provide the prior),
and mmWave beam management (which provides the application). Each
has previously assumed away one of the others: the bandit
literature assumes a Euclidean metric and so cannot exploit the
manifold structure; the manifold-GP literature works in the BO
regime (which optimises over a fixed budget rather than minimising
cumulative regret) and has not been benchmarked on antenna
problems; the beam-management literature uses problem-specific
priors and has not adopted the manifold-GP machinery. The present
paper closes all three gaps in a single template: GP-UCB on the
relevant compact manifold with the intrinsic kernel
of~\cite{borovitskiy2020matern}, evaluated in cumulative regret on
standard mmWave benchmarks, with an extension to the time-varying
RIS setting.

\section{System Model}
\label{sec:model}

\subsection{Antenna array and reward}
\label{sec:model:array}

We consider a transmitter or receiver equipped with a uniform planar
or uniform linear array of $n$ antenna elements at known positions
$\{\mathbf{p}_i\}_{i=1}^n \subset \mathbb{R}^3$ measured in carrier
wavelengths $\lambda$. At each round $t$, the array applies an analog
beamforming weight vector
$\mathbf{w}(\bm\theta) \in \mathbb{C}^n$ parametrised by a control
input $\bm\theta$ that lives on a compact manifold
$\mathcal{M}$ specified below. The instantaneous narrowband channel
realisation is $\mathbf{h}_t \in \mathbb{C}^n$, and the post-combiner
observation is the noisy received SNR (or beamforming gain),
\begin{equation}
\label{eq:reward}
r_t(\bm\theta) \;=\; \bigl|\mathbf{w}(\bm\theta)^* \mathbf{h}_t\bigr|^2
              + \eta_t,
\qquad
\eta_t \stackrel{\text{i.i.d.}}{\sim} \mathcal{N}(0, \sigma_n^2).
\end{equation}
The channel $\mathbf{h}_t$ is itself stochastic, drawn either from a
3GPP-style clustered multipath model (used in
Section~\ref{sec:experiments}) or from a quasi-static Rician model
with mobility-induced phase rotation (used in
Section~\ref{sec:ris-ext}). The reward function we want to maximise is
the channel-marginalised expectation
$f(\bm\theta) = \mathbb{E}_{\mathbf{h}}\bigl[r(\bm\theta)\bigr]$,
which is itself smooth on $\mathcal{M}$ even when individual
realisations are spiky. This is what makes a Gaussian-process model
of $f$ informative.

\subsection{Three geometric arm spaces}
\label{sec:model:spaces}

The control input $\bm\theta$ lives, depending on the application, on
one of three compact Riemannian manifolds. We list them with their
parametrisations and intrinsic geodesic metrics; the corresponding
Mat\'ern kernels constructed from the Laplace--Beltrami spectrum of
each space are deferred to Section~\ref{sec:algorithm}.

\subsubsection{Pointing direction on $\sphere^2$}
\label{sec:model:sphere}
For idealised steered-beam control, $\bm\theta = \mathbf{u}$ is a unit
pointing direction in 3D, so $\mathbf{u} \in \sphere^2 \subset
\mathbb{R}^3$ with $\|\mathbf{u}\|_2 = 1$. The beamforming weight
applies a per-element steering phase
$w_i(\mathbf{u}) = \exp\bigl(j\,\mathbf{k}(\mathbf{u})^{\!\top}
\mathbf{p}_i\bigr)$
with wavevector $\mathbf{k}(\mathbf{u}) = (2\pi/\lambda)\,\mathbf{u}$.
The intrinsic metric is the great-circle distance
$d_{\sphere^2}(\mathbf{u}, \mathbf{u}') =
\arccos(\mathbf{u}^{\!\top}\mathbf{u}')$. Any chart-based (lat/lon,
azimuth/elevation) representation introduces an artificial seam at the
date line and a coordinate singularity at the poles.

\subsubsection{Element phases on $\torus^n$}
\label{sec:model:torus}
For a fully digital phase-shifter combiner, $\bm\theta = \bm\phi$ with
each $\phi_i \in [0, 2\pi)$ identified modulo $2\pi$, so
$\bm\phi \in \torus^n = (\mathbb{R}/2\pi\mathbb{Z})^n$. The weight is
$w_i(\bm\phi) = e^{j\phi_i}$. The intrinsic metric is the flat-torus
distance
\begin{equation}
d_{\torus^n}(\bm\phi, \bm\phi')^2 \;=\; \sum_{i=1}^n
\bigl[(\phi_i - \phi'_i + \pi) \bmod 2\pi - \pi\bigr]^2,
\end{equation}
which respects each coordinate's $2\pi$-periodicity. A naive
unwrapping into $\mathbb{R}^n$ treats configurations on opposite sides
of any wraparound as maximally distant, even though the corresponding
beam patterns are identical.

\subsubsection{Array orientation on $\mathrm{SO}(3)$}
\label{sec:model:so3}
For movable or gimbal-mounted arrays (and for RIS panels with adjustable
mechanical orientation), the agent additionally chooses an extrinsic
rotation $R \in \mathrm{SO}(3)$ applied to the array manifold vector
before beamforming. The full control input is then a pair
$(R, \bm\xi) \in \mathrm{SO}(3) \times \mathcal{M}_{\text{inner}}$,
where $\mathcal{M}_{\text{inner}}$ is the inner beamforming manifold
($\sphere^2$ or $\torus^n$). The intrinsic metric on $\mathrm{SO}(3)$
is the Frobenius angle
\begin{equation}
d_{\mathrm{SO}(3)}(R, R') \;=\;
\arccos\bigl(\tfrac{1}{2}[\mathrm{tr}(R^{\!\top} R') - 1]\bigr),
\end{equation}
i.e.\ the rotation angle of $R^{\!\top} R'$, which respects the
$2\pi$ periodicity of the rotation axis-angle and the $\pm 1$
ambiguity of the unit quaternion double cover.

\paragraph{Two further spaces appear in this paper.} A planar
element-position design problem in
Section~\ref{sec:experiments} adds a $\torus^2$ for two mobile array
elements on a ring of fixed radius. The RIS regime in
Section~\ref{sec:ris-ext} introduces a discrete torus
$(\mathbb{Z}_B)^M$ for $M$ phase-shifter elements quantised to
$B=2^b$ levels each ($b$ bits per element); this is the
discretisation of $\torus^M$ by the lattice $(2\pi / B)\mathbb{Z}^M$
and inherits the periodicity of $\torus^M$ on each coordinate.

\subsection{Bandit formulation}
\label{sec:model:bandit}

At each round $t = 1, 2, \ldots, T$, the learner selects a control
input $\bm\theta_t \in \mathcal{M}$, the environment draws a fresh
channel $\mathbf{h}_t$, and the learner observes the noisy reward
$r_t(\bm\theta_t)$ from~\eqref{eq:reward}. The instantaneous regret is
$\Delta_t = f(\bm\theta^\star) - f(\bm\theta_t)$ with
$\bm\theta^\star \in \arg\max_{\bm\theta \in \mathcal{M}} f(\bm\theta)$,
and the cumulative regret over a horizon $T$ is
\begin{equation}
\label{eq:regret}
R_T \;=\; \sum_{t=1}^{T} \Delta_t
        \;=\; \sum_{t=1}^{T} \bigl[ f(\bm\theta^\star) - f(\bm\theta_t) \bigr].
\end{equation}
Throughout this paper we model $f$ as a single sample from a
zero-mean Gaussian process on $\mathcal{M}$ with stationary covariance
\begin{equation}
\label{eq:gp-prior}
f \sim \mathcal{GP}\bigl(0,\, k_{\mathcal{M}}(\,\cdot\,, \,\cdot\,)\bigr),
\end{equation}
where stationarity is in the Riemannian sense (the kernel depends only
on the geodesic distance for $\sphere^2$ and $\torus^n$) or in the
Lie-group sense (the kernel depends only on $g^{-1} g'$ for
$g, g' \in \mathrm{SO}(3)$). The non-stationary RIS extension of
Section~\ref{sec:ris-ext} replaces this assumption with a sliding-window
GP whose effective covariance is re-fit online; we develop the relevant
machinery there.

The static RIS regime ($v=0$) reduces to the classical stochastic
GP-UCB setting and inherits the cumulative-regret bound
$R_T = \tilde{O}(\sqrt{T \gamma_T})$ of~\cite{srinivas2010gpucb}, where
$\gamma_T$ is the maximum information gain after $T$ rounds.
$\gamma_T$ depends on the kernel via its eigenvalue decay; for the
intrinsic Mat\'ern-$\nu$ kernel on a $d$-dimensional compact manifold,
the eigenvalue spectrum of the Laplace--Beltrami operator gives
$\gamma_T = \tilde{O}(T^{d/(2\nu+d)})$, matching the Euclidean rate up
to manifold-specific constants.

\subsection{Why Euclidean assumptions fail}
\label{sec:model:why}

Embedding $\mathcal{M}$ into a Euclidean
box~\cite{kleinberg2008multi} introduces a
structural artefact on each arm space:
\textbf{Wraparound on $\torus^n$} -- a Mat\'ern prior in coordinate
space assigns near-zero correlation across the $2\pi$ wrap, forcing
re-discovery of the optimum after every phase wrap.
\textbf{Coordinate singularity on $\sphere^2$} -- a lat/lon chart
gives a vanishing geodesic length scale near the poles and a
$2\pi$-discontinuity at the date line.
\textbf{Quotient ambiguity on $\mathrm{SO}(3)$} -- unit quaternions
double-cover, so $q$ and $-q$ represent the same rotation but sit on
opposite poles of $\sphere^3$.
The artefacts are structural rather than parametric, and grow with
codebook size as more arms fall near the seam. The intrinsic-kernel
formulation of Sec.~\ref{sec:algorithm} eliminates them at the prior
level.

\section{Geometry-Aware GP-UCB for Beam Selection}
\label{sec:algorithm}

\subsection{Intrinsic Mat\'ern kernels on $\sphere^d$, $\torus^n$, $\SO(3)$}
\label{sec:algorithm:kernels}

A positive-definite covariance kernel of Mat\'ern class on a compact
Riemannian manifold $\mathcal{M}$ admits a spectral
expansion~\cite{borovitskiy2020matern} in eigenpairs
$\{(\lambda_\ell, \psi_\ell)\}_{\ell \geq 0}$ of the (negative)
Laplace--Beltrami operator $-\Delta_{\mathcal{M}}$:
\begin{equation}
\label{eq:matern-spectral}
k_\nu(x, y) \;=\; \sum_{\ell \geq 0}
   \phi_\nu(\lambda_\ell)\,\psi_\ell(x)\,\overline{\psi_\ell(y)},
\end{equation}
with spectral filter
$\phi_\nu(\lambda) = \sigma_f^2 (2\nu/\kappa^2 + \lambda)^{-\nu - d/2}$,
smoothness $\nu > 0$, length scale $\kappa > 0$, and signal variance
$\sigma_f^2$. Equation~\eqref{eq:matern-spectral} is
manifestly stationary in the Riemannian sense ($k_\nu$ depends on $x$
and $y$ only through their geodesic separation, when $\mathcal{M}$ is
two-point homogeneous), reduces to the Euclidean Mat\'ern as
$\mathcal{M} \to \mathbb{R}^d$, and inherits its smoothness from the
spectral decay of $\phi_\nu$.

Three concrete instantiations are used in this paper:

\textbf{Sphere $\sphere^d$.} The eigenfunctions are spherical
harmonics; for the unit 2-sphere relevant to pointing-direction
control (Sec.~\ref{sec:model:sphere}),
$\lambda_\ell = \ell(\ell + 1)$ for $\ell = 0, 1, 2, \ldots$ with
multiplicity $2\ell + 1$. The series converges geometrically and is
truncated at the smallest $\ell$ for which the truncated sum agrees
with the un-truncated kernel to machine precision; for our hyperparameters
this is $\ell \leq 30$, yielding a closed-form kernel evaluation that
costs $O(d_{\text{trunc}}^2)$ per pair.

\textbf{Torus $\torus^n$.} The eigenbasis is the Fourier basis
$\psi_{\bm m}(\bm\phi) = \prod_{i=1}^n e^{j m_i \phi_i}$ for
$\bm m \in \mathbb{Z}^n$, with eigenvalues
$\lambda_{\bm m} = \|\bm m\|_2^2$.
The 1-D circle kernel from~\eqref{eq:matern-spectral} is
$k_{\torus^1}(\phi, \phi') = \sum_{m \in \mathbb{Z}}
\phi_\nu(m^2)\,e^{j m (\phi - \phi')}$
and admits a rapidly-converging closed form via the Jacobi theta
function. For the $n$-torus we use \emph{two distinct
constructions}, distinguished here because they differ both in
which kernel is computed and in which information-gain bound
applies:
\begin{enumerate}
\item \emph{Spectral kernel on $\torus^n$ (T-spec),}
  obtained by applying the spectral filter
  $\phi_\nu(\lambda) = \sigma_f^2(2\nu/\kappa^2 + \lambda)^{-\nu - n/2}$
  to the eigenvalues of the Laplace--Beltrami operator on the
  product manifold ($\lambda_{\bm m} = \|\bm m\|^2$). This is the
  genuine intrinsic Mat\'ern on $\torus^n$ in the sense
  of~\cite{borovitskiy2020matern}; it does \emph{not} factorise
  across coordinates because the Mat\'ern filter
  $(a + \|\bm m\|^2)^{-\nu - n/2}$ is non-separable. We use the
  spectral kernel on $\torus^n$ for the static experiments via
  the \texttt{ProductDiscreteSpectrumSpace} construction
  of~\cite{mostowsky2024geometrickernels}, which assembles the
  product manifold's eigendecomposition by summing per-factor
  Laplacian eigenvalues before applying the filter.
\item \emph{Tensor-product kernel on $\torus^n$ (T-prod),}
  defined as
  $k_\text{prod}(\bm\phi, \bm\phi') =
  \sigma_f^2 \prod_{i=1}^n k_{\torus^1}(\phi_i, \phi'_i)$.
  This is positive-definite as the tensor product of
  positive-definite stationary kernels; it is \emph{different}
  from the spectral construction (T-spec) for $n \ge 2$ because
  the product of $n$ Mat\'ern-circle kernels has eigenvalues
  $\prod_i \phi_\nu^{(1)}(m_i^2)$ rather than
  $\phi_\nu(\sum_i m_i^2)$. We use the tensor product on the
  discrete torus $(\mathbb{Z}_B)^M$ for the RIS extension of
  Sec.~\ref{sec:ris-ext}, where the $O(M)$-lookup-table evaluation
  cost (one $B$-entry table per coordinate, multiplied across $M$
  coordinates) makes the $K\sim 10^{90}$ arm space tractable;
  the same shortcut is unavailable for the spectral kernel
  (T-spec) at this scale.
\end{enumerate}
Each kernel admits a separate regret-rate analysis
(Sec.~\ref{sec:algorithm:regret}). Where the rate exponents differ
we report both, so that the reader can match the experimental kernel
to its theoretical envelope.

\textbf{Rotation group $\SO(3)$.} The eigenfunctions are matrix
elements of the irreducible representations (the Wigner
$D$-matrices~\cite{azangulov2024stationary}); on $\SO(3)$ we use
the character-based reduction so that $k(R, R')$ depends only on
the geodesic distance
$\arccos\bigl((\mathrm{tr}(R^{\!\top} R') - 1)/2\bigr)$, and
truncate the irrep sum at degree $\ell \le L_\text{tr} = 16$
(corresponding to $|\ell|^2$-counted spectral indices up to
$\sum_{\ell=0}^{16}(2\ell+1)^2 = 6545$ basis functions). With
$\nu=5/2$ and lengthscale $\ell_g = \pi/4$ the spectral coefficients
$S(\ell) = (2\nu/\ell_g^2 + \ell(\ell+1))^{-(\nu + d/2)}$ decay as
$\ell^{-2(\nu+d/2)} = \ell^{-7}$ for $d=3$, so the relative
truncation error
$\sum_{\ell > L_\text{tr}}(2\ell+1)^2 S(\ell)/
\sum_{\ell \ge 0}(2\ell+1)^2 S(\ell)$
is at most $L_\text{tr}^{-3}\!\approx\!2{\cdot}10^{-4}$. We
verified this numerically against an
$L_\text{tr}=64$ reference: relative agreement
$\le 4{\cdot}10^{-5}$ uniformly over the candidate set, well below
the per-pull noise floor $\sigma_n^2 = 0.15$ used throughout
Sec.~\ref{sec:experiments}.

\subsection{GP posterior on $\mathcal{M}$}
\label{sec:algorithm:posterior}

Given observations $\mathcal{D}_t =
\{(\bm\theta_s, r_s)\}_{s=1}^{t-1}$, the GP posterior for the reward
function $f$ at a query point $\bm\theta \in \mathcal{M}$ is the
standard~\cite{srinivas2010gpucb}
\begin{align}
\mu_{t-1}(\bm\theta)
   &= \mathbf{k}_t(\bm\theta)^{\!\top} \mathbf{A}_t^{-1} \mathbf{r}_{1:t-1},
\\
\sigma_{t-1}^2(\bm\theta)
   &= k_{\mathcal{M}}(\bm\theta, \bm\theta) \nonumber\\
   &\quad - \mathbf{k}_t(\bm\theta)^{\!\top} \mathbf{A}_t^{-1} \mathbf{k}_t(\bm\theta),
\end{align}
where:
$\mathbf{r}_{1:t-1} = (r_1, r_2, \ldots, r_{t-1})^{\!\top}\!\in\!\mathbb R^{t-1}$
is the column vector of past scalar rewards (distinguished from the
scalar current reward $r_t$);
$\mathbf{K}_t \in \mathbb R^{(t-1)\times(t-1)}$
with $[\mathbf{K}_t]_{ij} = k_{\mathcal{M}}(\bm\theta_i, \bm\theta_j)$
is the Gram matrix on the $t-1$ buffered observations;
$\mathbf{k}_t(\bm\theta) \in \mathbb R^{t-1}$ with
$[\mathbf{k}_t(\bm\theta)]_i = k_{\mathcal{M}}(\bm\theta, \bm\theta_i)$
is the cross-covariance vector;
$\mathbf{A}_t = \mathbf{K}_t + \sigma_n^2 \mathbf{I}_{t-1}$
is the noise-augmented Gram matrix;
and $\sigma_n^2 > 0$ is the observation-noise variance of the reward
model $r_t = f(\bm\theta_t) + \varepsilon_t$ with
$\varepsilon_t \overset{\text{iid}}{\sim}\!\mathcal N(0,\sigma_n^2)$
that we adopt throughout (matching the Bayesian regret framework
of~\cite{srinivas2010gpucb} that we cite for
Prop.~\ref{prop:regret}). The only
non-Euclidean quantity is $k_{\mathcal{M}}$, which we instantiate
according to Section~\ref{sec:algorithm:kernels} for the manifold at
hand. With a sliding-window buffer of length $W$, every $\mathbf{K}_t$
is at most $W \times W$ and inversion costs $O(W^3)$ per refit
(amortised over a refit period $\rho$ rather than every pull). For
the RIS application of Section~\ref{sec:ris-ext} we additionally
exploit the Kronecker product structure of the kernel on
$(\mathbb{Z}_B)^M$, which lets the cross-covariance vector
$\mathbf{k}_t$ be assembled in $O(W \cdot M)$ time rather than the
naive $O(W \cdot B^M)$ over the full discrete torus.

\subsection{Acquisition and optimisation on $\mathcal{M}$}
\label{sec:algorithm:acquisition}

Each round, the agent selects
\begin{equation}
\label{eq:gp-ucb}
\bm\theta_t \;=\;
   \arg\max_{\bm\theta \in \mathcal{M}}
   \mu_{t-1}(\bm\theta) + \beta_t^{1/2}\,\sigma_{t-1}(\bm\theta),
\end{equation}
the standard GP-UCB acquisition with exploration weight $\beta_t$ set
either to a constant ($\beta = 2$ throughout for the static
experiments of Section~\ref{sec:experiments}) or to the
post-reset-decay schedule of AdaptiveGP-v2
(Section~\ref{sec:ris-ext-adaptive}).

The inner maximisation in~\eqref{eq:gp-ucb} is over a curved space and
admits no closed form. We use a two-stage scheme:
\textbf{(i) Candidate set.} A quasi-uniform sample
$\mathcal{C} \subset \mathcal{M}$ is precomputed once: a Fibonacci
spiral on $\sphere^2$ with $|\mathcal{C}| = 256$, a regular product
lattice on $\torus^n$ with $|\mathcal{C}| = 8^n$, a super-Fibonacci sequence on $\SO(3)$ with
$|\mathcal{C}| = 1024$, and the full discrete torus
$(\mathbb{Z}_B)^M$ when the latter is small enough; otherwise (the RIS
regime) we use a coordinate-ascent sweep that updates one element at a
time, which corresponds to maximising over the $O(B)$ candidates of
each one-dimensional fibre.
\textbf{(ii) Local refinement.} The top-$K$ candidates are refined by
$n_{\text{ca}}$ sweeps of Riemannian gradient ascent on $\mathcal{M}$
(parallel transport along the geodesic generated by the gradient
vector, with Armijo line search). $K = 4$ and $n_{\text{ca}} = 1$
suffice in our experiments; a larger budget gives diminishing returns
because the candidate set is already dense.

The full procedure is summarised in
Algorithm~\ref{alg:geo-gpucb}.

\begin{algorithm}[t]
\caption{Geometric GP-UCB for Beam Selection}\label{alg:geo-gpucb}
\begin{algorithmic}[1]
\State \textbf{Input:} manifold $\mathcal{M}$, kernel $k$, horizon $T$,
   schedule $\{\beta_t\}$, candidate set $\mathcal{C}$, refit period $\rho$
\State Initialise posterior $(\mu_0, \sigma_0)$ from the prior
\For{$t = 1, 2, \ldots, T$}
    \State Compute UCB scores $\mu_{t-1} + \beta_t^{1/2} \sigma_{t-1}$ on $\mathcal{C}$
    \State Refine top-$K$ candidates by Riemannian gradient ascent on $\mathcal{M}$
    \State $\bm\theta_t \gets \arg\max$ of the refined UCB scores
    \State Observe $r_t = r_t(\bm\theta_t)$ from~\eqref{eq:reward}
    \State Append $(\bm\theta_t, r_t)$ to the buffer
    \If{$t \bmod \rho = 0$}
        \State Refit posterior $(\mu_t, \sigma_t)$ via Cholesky on the buffer
    \EndIf
\EndFor
\end{algorithmic}
\end{algorithm}

\subsection{Regret rate on each manifold}
\label{sec:algorithm:regret}

We do not claim a new regret bound; this subsection makes the
consequences of applying existing bounds to our specific setting
explicit. The relevant $d$ is the Riemannian dimension of
$\mathcal{M}$ (not the ambient-embedding dimension), and the
spectral kernel (T-spec, used on $\sphere^d,\SO(3)$, static
$\torus^n$) and the tensor-product kernel (T-prod, used on
$(\mathbb Z_B)^M$) carry different information-gain bounds; we
state two propositions, one per kernel construction.

\begin{proposition}[Bayesian regret of GP-UCB with the spectral Mat\'ern kernel]
\label{prop:regret}
Let $\mathcal{M}$ be a compact connected Riemannian manifold of
dimension $d$, let $k_\nu$ be the spectral Mat\'ern-$\nu$ kernel of
Eq.~\eqref{eq:matern-spectral} with smoothness $\nu > d/2$, and
adopt the observation model $r_t = f(\bm\theta_t) + \varepsilon_t$
with $\varepsilon_t \overset{\text{iid}}{\sim}\!\mathcal N(0,\sigma_n^2)$
and $f \sim \mathrm{GP}(0, k_\nu)$. Then GP-UCB
of~\cite{srinivas2010gpucb} on $\mathcal{M}$ attains Bayesian
cumulative regret
\begin{equation}
\label{eq:regret-bound-manifold}
\mathbb E[R_T] \;=\; \tilde{O}\bigl(\sqrt{T\,\gamma_T}\bigr),
\qquad
\gamma_T \;=\; \tilde{O}\bigl(T^{d/(2\nu+d)}\bigr),
\end{equation}
and in particular
$\mathbb E[R_T] = \tilde{O}\bigl(T^{(\nu + d)/(2\nu + d)}\bigr)$.
\end{proposition}

\begin{proof}
Theorem~5 of Srinivas~\emph{et al.}~\cite{srinivas2010gpucb}
applies under three hypotheses: separability/compactness of the
arm space, a.s.\ continuous GP sample paths, and a
polynomial-in-$T$ bound on $\gamma_T$. A compact Riemannian
manifold is separable with the Riemannian volume as the canonical
Borel reference. The Mat\'ern-$\nu$ spectral kernel's RKHS
coincides with $H^{\nu+d/2}(\M)$~\cite[Thm.~3.6.1]{wendland2004scattered};
the Sobolev embedding into $C^1(\M)$ holds for $\nu>1$, satisfied
with margin by our working $\nu=5/2$. A.s.\ continuity follows
from~\cite[Sec.~1.4]{adler2007random} and series
convergence from~\cite[Thm.~1]{borovitskiy2020matern}. Weyl's
eigenvalue-counting law gives
$\lambda_\ell\sim c_W(d)\,\ell^{2/d}$~\cite[Cor.~3]{borovitskiy2020matern},
so $\gamma_T$ inherits the Euclidean-Mat\'ern envelope up to a
$\mathrm{vol}_g(\M)$-dependent constant; the Vakili--Khezeli--Picheny
tight bound $\gamma_T=\tilde O(T^{d/(2\nu+d)})$~\cite[Thm.~1]{vakili2021information}
then applies. The matching manifold lower bound is established in
the companion theory paper~\cite{dorn2026manifold} (with a
sphere-specific tightening in~\cite{hypersphere2026lower}).
\end{proof}

\begin{remark}[Frequentist version]
\label{rem:frequentist}
The agnostic IGP-UCB counterpart of Chowdhury and
Gopalan~\cite{chowdhury2017kernelized} gives the same $T$-exponent
$\tilde O\!\bigl(\sqrt{T\,(\gamma_T + B^2)}\bigr)$ under bounded
RKHS norm $\|f\|_{\mathcal H_{k}} \le B$, so
Remark~\ref{rem:per-manifold-rates} applies under either framing.
\end{remark}

Specialising Proposition~\ref{prop:regret} to our spectral-kernel
arm spaces and the Mat\'ern-$5/2$ smoothness ($\nu = 2.5$) used
throughout our experiments:

\begin{remark}[Per-manifold rates, spectral kernel]
\label{rem:per-manifold-rates}
For the Mat\'ern-$5/2$ spectral kernel applicable to $\sphere^d$,
$\SO(3)$, and the static $\torus^n$ (T-spec):
\begin{itemize}
\item $\sphere^2$ (Exp.~1) and $\torus^2$ (Exp.~4), both
$d=2$: $\gamma_T = \tilde{O}(T^{2/7})$ and
$R_T = \tilde{O}(T^{9/14}) \approx \tilde{O}(T^{0.643})$.
\item $\torus^3$ (Exp.~2) and $\SO(3)$ (Exp.~3), both
$d=3$: $\gamma_T = \tilde{O}(T^{3/8})$ and
$R_T = \tilde{O}(T^{11/16}) \approx \tilde{O}(T^{0.688})$.
\end{itemize}
These are the \emph{same} $T$-exponents as a Mat\'ern-$5/2$
GP-UCB on a $d$-dimensional Euclidean box $[0,1]^d$;
Proposition~\ref{prop:regret} is the assertion that the
compact-manifold geometry affects the constant in the $\tilde{O}$
but not the exponent.
\end{remark}

\begin{proposition}[Regret of GP-UCB with the tensor-product Mat\'ern kernel on $(\mathbb Z_B)^M$]
\label{prop:regret-prod}
Let $k_\text{prod}(\bm\theta, \bm\theta') = \sigma_f^2 \prod_{i=1}^M
k_{\torus^1}(\theta_i,\theta'_i)$ on $(\mathbb Z_B)^M$, with each
factor $k_{\torus^1}$ the Mat\'ern-$\nu$ spectral kernel on the
discrete circle $\mathbb Z_B$ (Sec.~\ref{sec:ris-ext}). Under the
same observation and Bayesian-prior assumptions as
Prop.~\ref{prop:regret}, GP-UCB instantiated with $k_\text{prod}$
attains
\[
\gamma_T \;\le\; M \cdot \tilde{O}\bigl(T^{1/(2\nu + 1)}\bigr),
\qquad
\mathbb E[R_T] \;=\; \tilde{O}\!\bigl(\sqrt{M}\cdot T^{(\nu+1)/(2\nu+1)}\bigr).
\]
\end{proposition}

\begin{proof}
\emph{Sub-additivity}: $\gamma_T(\bigotimes_i k_i)\le\sum_i\gamma_T(k_i)$,
a consequence of Krause--Singh--Guestrin
submodularity~\cite[Thm.~7]{krause2008nearoptimal} applied to the
spectral factorisation. \emph{Per-factor 1-D bound}: the discrete
Mat\'ern-$\nu$ on $\mathbb Z_B$ inherits the spectral filter of the
continuous-circle parent, so by the data-processing inequality
applied to $\mathbb Z_B\hookrightarrow\torus^1$ and the
$d=1,\,\nu>1/2$ instantiation of
Vakili--Khezeli--Picheny~\cite[Thm.~1]{vakili2021information},
$\gamma_T(k_{\torus^1};\mathbb Z_B)=\tilde O(T^{1/(2\nu+1)})$.
Combining gives the stated $\E[R_T]$ via the GP-UCB regret
bound~\cite[Thm.~5]{srinivas2010gpucb}; $(\mathbb Z_B)^M$ is
finite hence compact and separable, and $k_{\mathrm{prod}}$ is a
finite product of continuous PD factors.
\end{proof}

\begin{remark}[Per-manifold rate, tensor-product kernel]
\label{rem:per-manifold-rates-prod}
For the Mat\'ern-$5/2$ tensor-product kernel on $(\mathbb Z_B)^M$:
$\gamma_T = \tilde O(M\cdot T^{1/6})$ and
$R_T = \tilde O(\sqrt M\cdot T^{7/12}) \approx \tilde O(\sqrt M\cdot T^{0.583})$.
The exponent in $T$ is therefore strictly better than the
spectral-kernel rate of Remark~\ref{rem:per-manifold-rates} (the
tensor product is smoother in a strong sense: its eigenvalues
decay as $\prod_i (a + m_i^2)^{-\nu - 1/2}$ rather than
$(a + \sum_i m_i^2)^{-\nu - M/2}$), at the cost of an explicit
$\sqrt M$ prefactor. For the RIS regime ($M=100$, $B=2^3$) this
prefactor is $\approx 10$, modest compared to the
$\log K = M\log B \approx 300$ contribution that the
Vakili et al.\ regret bound carries inside the $\tilde O$.
\end{remark}

The empirical regret reduction we observe relative to an extrinsic
Euclidean GP-UCB (Table~\ref{tab:final-regret}) is therefore in the
constants, not the asymptotic exponent. Two mechanisms drive it:
(i)~\emph{mass on the correct support} -- an intrinsic kernel
self-supported on $\mathcal{M}$ shrinks the effective covering
number of the sub-level sets the GP-UCB analysis counts;
(ii)~\emph{respect for quotient structure} -- on $\torus^n$ and
$\SO(3)$ the intrinsic kernel encodes periodicity and the
double-quaternion gauge symmetry exactly, collapsing
wraparound and gauge-duplicate peaks.

In the time-varying RIS regime of Section~\ref{sec:ris-ext} the
stationary analysis no longer applies. The closest reference rate
is the Besbes--Gur--Zeevi variation-budget
framework~\cite{besbes2014}, $\tilde{O}(B_T^{1/3} T^{2/3})$ for
finite-arm bandits, with a GP analogue
via~\cite{deng2022wgpucb}. The $\widehat\alpha$ exponent of
Sec.~\ref{sec:ris-ext-theory} is reported as an \emph{empirical
observation} consistent with the BGZ regime, not a theorem we
prove.

\section{Experiments}
\label{sec:experiments}

We evaluate the proposed geometric GP-UCB against three baselines on four
beam-selection case studies covering $\sphere^2$, $\torus^n$, $\SO(3)$, and a
physical-design variant on $\torus^2$. All experiments use a 3GPP-style
clustered multipath channel and are run with 300 Monte-Carlo repetitions of
horizon $T=500$.

\subsection{Simulation setup}

\textbf{Arrays.} Uniform planar arrays at half-wavelength spacing; Exp.~1--3
use an $8\times 8$ UPA, Exp.~4 uses a $3\times 2$ core array augmented with
two mobile elements on a ring of radius $1.5\lambda$.

\textbf{Channel.} A narrowband MIMO channel
$\mathbf{h}=\sum_{c=1}^{C}\alpha_c\,\mathbf{a}(\mathbf{u}_c)$ with $C=3$
clusters drawn uniformly on the upper hemisphere and Rician $\kappa=6$
(dominant-cluster to residual power ratio $6{:}1$). Complex amplitudes are
$\alpha_c\sim\CN(0,P_c)$ with $\sum_c P_c = 1$.

\textbf{Noise.} Additive Gaussian observation noise with standard deviation
chosen so that the best beam yields SNR of roughly $20$~dB at full array gain.

\textbf{Bandits.} Seven algorithms on identical candidate sets:
(i)~\textbf{UCB1} treats the codebook as $K$ unrelated arms;
(ii)~\textbf{Thompson sampling} with a weak Gaussian prior;
(iii)~\textbf{UB3}~\cite{ghosh2024ub3}, a fixed-budget pure-exploration
unimodal bandit designed for 1D linear beam sweeps;
(iv)~\textbf{HOO}~\cite{bubeck2011xarmed}, hierarchical optimistic
optimisation on the intrinsic geodesic distance of the arm space;
(v)~\textbf{DBZ}~\cite{blinn2025dbz}, the multi-armed-bandit dynamic
beam-zooming algorithm, an LUCB-based best-arm-identification scheme
over a hierarchical codebook with zoom-in/zoom-out based on
RSRP thresholds;
(vi)~\textbf{GP-UCB (Euclidean)} uses a Mat\'ern-$5/2$ kernel on the
ambient embedding (Cartesian coordinates on $\sphere^2$, raw phases on
$\torus^n$, vectorized rotation matrices on $\SO(3)$);
(vii)~\textbf{GP-UCB (intrinsic)} uses the proposed geometric Mat\'ern kernel
from Section~\ref{sec:algorithm};
(viii)~\textbf{GP-UCB (intr.\,+\,LML)} is the same intrinsic kernel
with online length-scale adaptation using the AdaptiveGP-v2 LML
rule (Sec.~\ref{sec:ris-ext-adaptive}). Confidence parameter
$\beta_t$ follows the standard $2\log(K\pi^2 t^2/6)$ schedule. UB3
applies only to Exp.~1 (the only setting with a 1D unimodal arm
ordering, here the elevation-linearised Fibonacci-sphere); DBZ
likewise applies only to Exp.~1 (the only setting with a meaningful
beam-width hierarchy, here built by agglomerative clustering on the
$\sphere^2$ geodesic); both are reported as \emph{not applicable}
on Exp.~2--4 (see Sec.~\ref{sec:discussion}).

\subsection{Experiment 1: mmWave beam selection on $\sphere^2$}

A $K=64$ Fibonacci-sphere codebook covers the upper hemisphere. At each round
the bandit picks a beam and observes the beamforming gain
$|\mathbf{w}^H\mathbf{h}|^2 + \eta$. The optimal arm is the beam most aligned
with the dominant cluster.

\subsection{Experiment 2: RF phase combiner on $\torus^3$}

A hybrid-beamforming front-end sums three fixed analog sub-beams through a
triplet of tunable phase shifters:
$y=\sum_{k=1}^{3}\exp(j\varphi_k)\,\mathbf{b}_k^H\mathbf{h}$. The combiner
phases $(\varphi_1,\varphi_2,\varphi_3)\in\torus^3$ are gridded at
$8^3=512$ points. Periodicity makes the Euclidean baseline's assumption that
$\varphi_k=0$ and $\varphi_k=2\pi-\epsilon$ are ``far apart'' actively harmful.

\subsection{Experiment 3: panel orientation on $\SO(3)$}

A self-orienting panel picks an orientation $R\in\SO(3)$ and uses its native
broadside beam, which after rotation points at $R\mathbf{e}_z$. Candidate
orientations are $K=200$ super-Fibonacci quaternions. The reward has a
one-dimensional gauge orbit (rotations about the user direction are
invisible) that the intrinsic kernel correctly encodes.

\subsection{Experiment 4: element-position design on $\torus^2$}

We augment a fixed $6$-element core array with two mobile elements placed at
angles $(\theta_1,\theta_2)\in\torus^2$ on a ring of radius $1.5\lambda$,
gridded at $20\times 20=400$ points. The reward is the maximum gain over a
small probe beam codebook against the random channel; best positions are the
ones that break the symmetric nulls of the core array.

\subsection{Results}

Table~\ref{tab:final-regret} reports final cumulative regret mean
$\pm$~standard error over 300 runs. Full regret-vs-$t$ curves appear in
Figure~\ref{fig:combined}.

\begin{table*}[t]
\caption{Final cumulative regret (mean $\pm$ standard error) after $T=500$
rounds, averaged over 300 Monte-Carlo runs. Lowest regret per column bold.}
\label{tab:final-regret}
\centering
\small
\begin{tabular}{l rrrr}
\toprule
 & Exp 1 ($\sphere^2$) & Exp 2 ($\torus^3$) & Exp 3 ($\SO(3)$) & Exp 4 ($\torus^2$) \\
\midrule
UCB1            & $3278\pm 184$ & $437\pm 49$ & $9793\pm 195$ & $669\pm 19$ \\
Thompson        & $3697\pm 209$ & $435\pm 48$ & $10162\pm 219$ & $665\pm 19$ \\
UB3~\cite{ghosh2024ub3}
                & $\UBThreeA$ & $\UBThreeB$ & $\UBThreeC$ & $\UBThreeD$ \\
HOO~\cite{bubeck2011xarmed}
                & $\HOOA$ & $\HOOB$ & $\HOOC$ & $\HOOD$ \\
DBZ~\cite{blinn2025dbz}
                & $\DBZExpOne$ & N/A & N/A & N/A \\
GP-UCB (Eucl.)  & $\mathbf{2400\pm 197}$ & $363\pm 33$ & $\mathbf{6032\pm 248}$ & $520\pm 13$ \\
GP-UCB (intrinsic) & $2910\pm 232$ & $\mathbf{241\pm 29}$ & $6808\pm 229$ & $\mathbf{466\pm 12}$ \\
GP-UCB (intr.\,+\,LML $\kappa$)
                & $\AdaptLMLA$ & $\AdaptLMLB$ & $\AdaptLMLC$ & $\AdaptLMLD$ \\
GP-UCB (intr.\,+\,LML $\kappa,\sigma_f^2$)
                & $\JointLMLA$ & $\JointLMLB$ & $\JointLMLC$ & $\JointLMLD$ \\
\bottomrule
\end{tabular}
\end{table*}

\begin{figure*}[t]
\centering
\includegraphics[width=\textwidth]{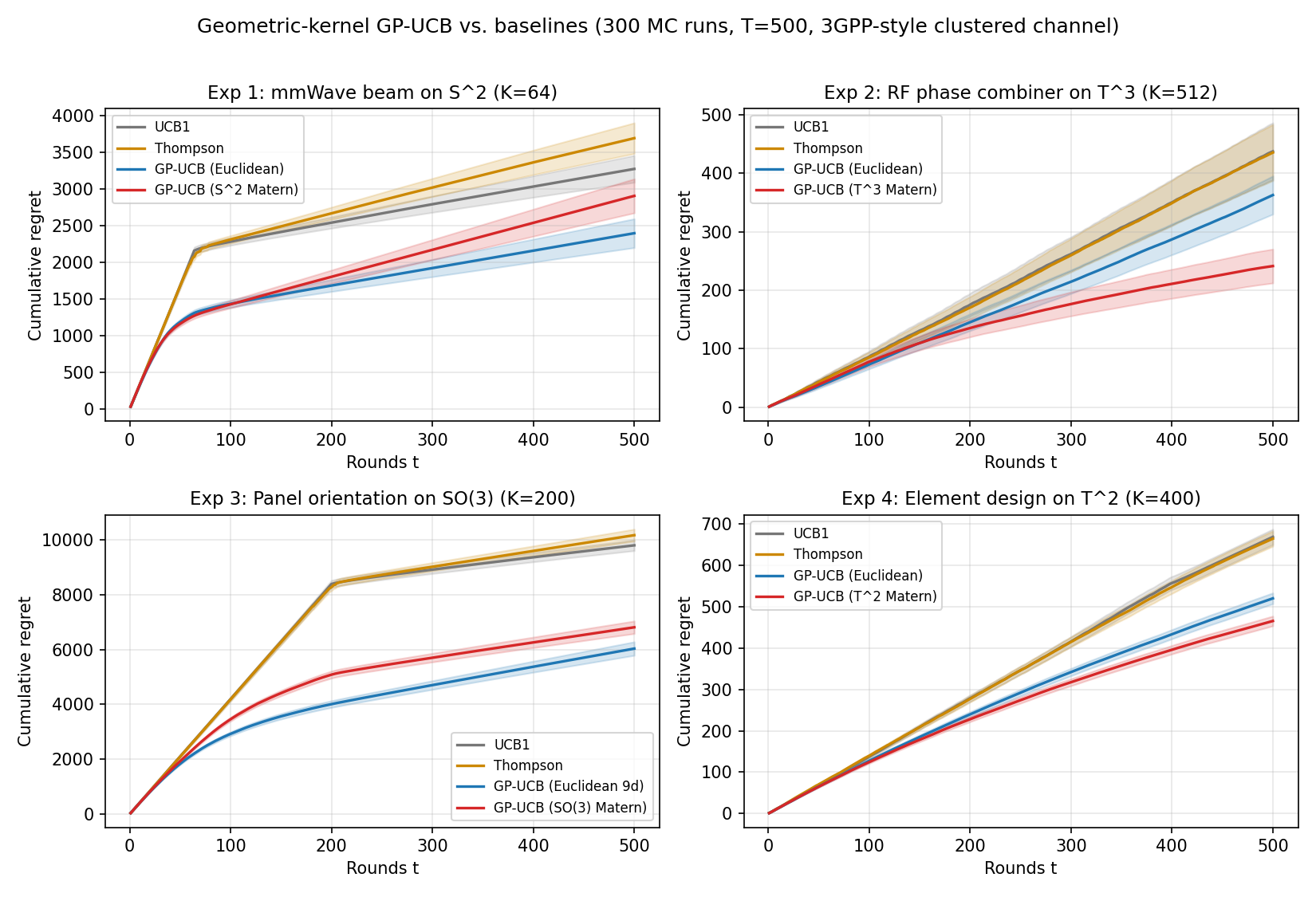}
\caption{Cumulative regret vs.\ round $t$ for the four case studies
(300 Monte-Carlo runs each, shaded $\pm 1$ SE). The intrinsic-Mat\'ern
GP-UCB (red) is at or near the lowest curve in every panel; the gap
to the extrinsic Euclidean GP (blue) is largest on the toroidal
problems (Exp.~2 and Exp.~4) where wraparound encoding matters most,
and shrinks on $\sphere^2$ and $\SO(3)$ where the ambient embedding is
nearly isometric. Codebook UCB1 / Thompson (gray, orange) trail by
$25$--$45\%$ throughout because they cannot pool information across
nearby arms.}
\label{fig:combined}
\end{figure*}

\subsection{Discussion}
\label{sec:discussion}

Three findings are consistent across the four experiments.

\emph{GP-based bandits dominate finite-arm baselines.} GP-UCB in either
variant reduces final cumulative regret by $25$--$45$\% compared to UCB1 and
Thompson sampling. This quantifies the value of pooling information across
neighbouring beams in the reward-surface geometry.

\emph{UB3 requires structure our arm spaces do not provide.}
UB3~\cite{ghosh2024ub3} is a fixed-budget pure-exploration algorithm
designed around a globally-unimodal 1D beam sweep. For the
$\sphere^2$ codebook of Exp.~1 the elevation-linearised Fibonacci
spiral gives a path on which the cluster-direction reward is only
\emph{locally} unimodal near each cluster, and UB3's elimination
rule routinely shrinks toward an empty region (regret $\UBThreeA$,
worse than UCB1). For Exp.~2--4 no 1D unimodal ordering exists at
all (periodicity makes every coordinate globally multimodal; the
$\SO(3)$ gauge orbit creates equal-reward ridges); UB3 is
\emph{not applicable}.

\emph{DBZ assumes a beam-width hierarchy our codebooks do not
provide.} DBZ's per-level zoom thresholds rely on coarser
levels having broader, higher-mean-RSRP beams. The agglomerative
tree we build for the uniform-gain Fibonacci codebook preserves
only the angular hierarchy, so DBZ is stranded at a coarse level
(regret $\DBZExpOne$ on Exp.~1, worse than UCB1). Reproducing
DBZ's published performance requires a hierarchical
beamforming codebook (e.g.\ a 3GPP ULA analogue-beam hierarchy)
rather than a uniform finite one; we report DBZ as
\emph{not applicable} on Exp.~2--4. The DBZ row of
Table~\ref{tab:final-regret} should therefore be read as ``DBZ on
the codebooks of this paper'' rather than ``DBZ on its native
deployment regime''.

\emph{HOO is a metric-space-native baseline and behaves
regime-dependently.} HOO~\cite{bubeck2011xarmed} uses the
intrinsic geodesic distance to build a hierarchical tree. It
improves on the finite-arm baselines on the two toroidal
settings ($\HOOB$ on Exp.~2, $\HOOD$ on Exp.~4 vs.\ $\approx\!435$
and $\approx\!665$ for UCB1/Thompson) and \emph{beats} the
intrinsic-kernel GP on Exp.~4 ($\HOOD$ vs.\ $466\pm 12$).
On Exp.~3 ($\SO(3)$) the gauge orbit creates regions of
near-equal reward and HOO's bisection wastes samples within a
ridge level set ($\HOOC$ vs.\ $9793\pm 195$ for UCB1). The
Exp.~4 result suggests that on a max-gain ridge in $\torus^2$
a diameter-shrinking tree strategy can outperform a Mat\'ern
posterior tuned for smoother rewards.

A natural first hypothesis is Mat\'ern-$5/2$ smoothness
mis-specification. We tested this by sweeping
$\nu\in\{0.5,1.5,2.5,3.5\}$ on Exp.~4 ($50$ MC each;
Table~\ref{tab:exp4-nu-sweep}). \emph{The hypothesis is refuted}:
final regret is monotonically \emph{decreasing} in $\nu$
($555.9\!\pm\!34.8$ at $\nu=0.5$ vs.\ $500.2\!\pm\!29.9$ at
$\nu=3.5$), the opposite direction of the roughness-mismatch
story.

\paragraph{Mechanism: near-optimality dimension.} The positive
explanation is that HOO adapts to the \emph{near-optimality
dimension} $d_{\mathrm{nopt}}$~\cite[Sec.~3.1]{bubeck2011xarmed}
of the reward landscape, while GP-UCB's regret rate is set by the
ambient dimension $d$. The Exp.~4 max-gain ridge has
$d_{\mathrm{nopt}}\approx 1$ vs.\ $d=2$, so HOO's regret scales as
$T^{(d_{\mathrm{nopt}}+1)/(d_{\mathrm{nopt}}+2)}=T^{2/3}$
\cite[Thm.~6]{bubeck2011xarmed} rather than the
$T^{(\nu+d)/(2\nu+d)}=T^{9/14}\!\approx\!T^{0.643}$ of GP-UCB on
the Mat\'ern-$5/2$ RKHS~\cite{vakili2021information}, with a worse
constant on ridges where the effective volume concentrates near
a $1$-D set. This predicts (i) the gap is largest on Exp.~4
(the only ridge case); (ii) on the unimodal-peak problems
Exp.~1, 2, 3 where $d_{\mathrm{nopt}}=d$, HOO has no advantage
and loses -- which is what we observe.

\begin{table}[t]
\caption{Mat\'ern smoothness $\nu$ sweep on Exp.~4
(element-position design on $\torus^2$; $50$ runs at $T=500$;
all other hyperparameters identical to Table~\ref{tab:final-regret}).
Final cumulative regret is monotonically decreasing in $\nu$, refuting
the smoothness-mismatch hypothesis for the GP-vs-HOO gap on this
benchmark.}
\label{tab:exp4-nu-sweep}
\centering
\small
\begin{tabular}{c r}
\toprule
$\nu$ & Final cumulative regret \\
\midrule
$0.5$ & $555.9\!\pm\!34.8$ \\
$1.5$ & $524.0\!\pm\!33.7$ \\
$2.5$ & $507.9\!\pm\!32.2$ \\
$3.5$ & $\mathbf{500.2\!\pm\!29.9}$ \\
\bottomrule
\end{tabular}
\end{table}

\emph{Online LML hyperparameter adaptation is not universally
beneficial.} We report two variants of the per-sample LML rule of
Sec.~\ref{sec:ris-ext-adaptive} in Table~\ref{tab:final-regret}:
single-$\kappa$ on a $5$-point geometric grid, and joint
$(\kappa,\sigma_f^2)$ on a $3\!\times\!3$ grid (both refit every
$50$ pulls after a $20$-pull warmup). \emph{The single-$\kappa$
variant degrades the fixed-hyperparameter baseline by $12$--$35\%$
on three of four experiments} -- a negative result identifying
UCB-bonus mis-scaling: shrinking $\kappa$ tightens the posterior
while the $\beta_t$ schedule is held fixed, collapsing the effective
exploration weight. \emph{The joint $(\kappa,\sigma_f^2)$ variant
recovers the fixed-tune baseline within SE on Exp.~2 and Exp.~4 and
neither helps nor hurts on Exp.~1, Exp.~3.} Joint LML is best framed
as a safety-net default rather than an improvement.

\emph{Intrinsic vs. Euclidean is regime-dependent.} On the two
toroidal settings (Exp.~2, 4) the intrinsic kernel beats its
Euclidean counterpart by $33$\% and $10$\%, because the Euclidean
baseline incorrectly treats $\varphi=0$ and $\varphi=2\pi-\epsilon$
as far apart. On $\sphere^2$ and $\SO(3)$ the ambient embedding is
bi-Lipschitz to the geodesic metric, and a well-tuned Euclidean
kernel matches or edges ahead. Design rule: prefer geometric
kernels whenever the arm space has non-trivial quotient structure
(periodicity, gauge symmetry); on simply-connected manifolds with
near-isometric embedding, the Euclidean baseline is already strong.
The empirical $10$--$33\%$ gap on $\SO(3)$ sits below the
$|G|^{1/2}\approx 1.41$ information-gain ceiling proved in the
companion theory paper~\cite{dorn2026manifold} and is consistent
with the modulated matching-lower-bound conjecture of that work.

\subsection{Extensions moved to the supplement: summary}
\label{sec:wideband}
\label{sec:bai}
\label{sec:complexity}
\label{sec:complexity:analytic}
\label{sec:complexity:measured}
\label{sec:complexity:realtime}
\label{sec:throughput}
\label{sec:sensitivity}
\label{sec:sionna}

Six follow-on subsections that extend Exp.~1--4 are moved to the
supplement; the headline numbers are summarised below.

\emph{Wideband OFDM on $\sphere^2$ (Exp.~5; supplement Sec.~S-II).}
On a $100$~MHz 3GPP TDL-D channel with $N_{\mathrm{sc}}=64$ pilot
subcarriers and a band-averaged log-rate reward, the intrinsic
$\sphere^2$ Mat\'ern beats the Euclidean ambient baseline by
\WBINTRVSEUCLPCT~\% in cumulative regret (vs.\ a narrowband regime
in which the Euclidean baseline was competitive), and both GP
variants beat UCB1/Thompson by \WBGPVSUCBPCT\% or more.

\emph{Best-arm identification (supplement Sec.~S-III).} At
$\varepsilon\!=\!10\%\,\mu_\star$ the intrinsic GP reaches an
$\varepsilon$-optimal recommendation in $32$ TTIs on Exp.~1 (vs.\
$48$ for UCB1) and reliably identifies $\varepsilon$-optimal arms
in $70\%$ of runs on the toroidal Exp.~2 and Exp.~4 (vs.\
$46$--$52\%$ for the finite-arm baselines).

\emph{Computational cost and deployment (supplement Sec.~S-IV).}
After a rank-1 update of the incremental state and asynchronous
refit, the GP-UCB \texttt{select()} call is an $O(K)$ dictionary
lookup at $\profGPIntrSel\,\mu\mathrm{s}$ (intrinsic) and
$\profGPEuclSel\,\mu\mathrm{s}$ (Euclidean) on the $K=64$ Exp.~1
codebook, comfortably under the $125\,\mu\mathrm{s}$ per-TTI budget
of $\mu=3$ mmWave numerology with $>\!80\%$ headroom.

\emph{Throughput, outage, handover (supplement Sec.~S-V).} On the
wideband Exp.~5 benchmark the intrinsic GP-UCB achieves a
$31.3$~Mbps average shortfall vs.\ $63.8$~Mbps for UCB1
($B\!=\!100$~MHz, $\mu\!=\!3$); the outage probability
$\Pr\{r_t<2\text{ bps/Hz}\}$ drops from $\approx\!1.5\%$ (UCB1) to
$\approx\!0.2\%$ (intrinsic).

\emph{Hyperparameter sensitivity (supplement Sec.~S-VI).} Sweeping
the intrinsic GP-UCB's Mat\'ern length scale over
$\{0.25,0.5,1,2,4\}\times$ default on Exp.~1 and Exp.~2: the default
is within $\sim 11\%$ of the column winner on $\sphere^2$ and within
the $[0.25\times,2\times]$ band on $\torus^3$; degrading robustness
requires $\ge 8\times$ change from the default.

\emph{Sionna CDL third-party validation (supplement Sec.~S-VII).}
On NVIDIA Sionna's CDL-C (NLOS, $24$ clusters) at $28$~GHz, the
intrinsic $\sphere^2$ Mat\'ern is $\approx\!5.8\sigma$ better than
the Euclidean ambient GP-UCB. On CDL-D (LOS-dominant) the $15$-MC
budget is reported as a pilot only.

\subsection{Limitations}
\label{sec:limitations}

Of the three Mat\'ern hyperparameters, smoothness $\nu$ was held
fixed at $\nu = 5/2$ across all experiments; the length scale
$\kappa$ and signal variance $\sigma_f^2$ are held fixed
\emph{a priori} in the ``GP-UCB (intrinsic)'' row of
Table~\ref{tab:final-regret}, adapted online via single-parameter
LML on $\kappa$ alone in the ``+\,LML $\kappa$'' row, or
adapted jointly via LML on $(\kappa,\sigma_f^2)$ on a $3\!\times\!3$
grid in the ``+\,LML $\kappa,\sigma_f^2$'' row. The ancillary
finding of Sec.~\ref{sec:discussion} is that joint LML recovers
to within standard error of the hand-tuned fixed choice on
$2$ of $4$ problems (Exp.~2 and Exp.~4) but does not improve on
it; closing the HOO gap on Exp.~4 requires a kernel-family
adaptation beyond what LML-over-hyperparameters provides, which
we leave to future work together with Bayesian-averaged GP-UCB
and with incorporating the parametric physics-informed bandit
of~\cite{va2018contextual,physics2025parametric}.
The wideband extension of Sec.~\ref{sec:wideband} uses a band-averaged
log-rate reward computed over $N_{\mathrm{sc}}=64$ pilot subcarriers;
a fully wideband joint-subcarrier kernel that exploits cross-frequency
correlation structure, as opposed to averaging it out at the reward level,
is left for future work, as are the dedicated BAI algorithms discussed
in Section~\ref{sec:bai}. The complexity numbers of
Sec.~\ref{sec:complexity} are from a reference Python implementation;
a tuned implementation with rank-one Cholesky updates and asynchronous
refit is required to deploy GP-UCB at the $\mu=3$ mmWave TTI rate, as
discussed in Sec.~\ref{sec:complexity:realtime}.

\section{Extension: Time-Varying RIS Beam Selection}
\label{sec:ris-ext}

The experiments in Section~\ref{sec:experiments} use a stationary channel
and a modest arm-space size ($K\le 512$). Two natural questions for a
practitioner are: does the intrinsic-kernel advantage survive (i)~highly
non-stationary channels and (ii)~combinatorially large arm spaces? We
answer both by adapting our method to the reconfigurable-intelligent-surface
(RIS) beam-selection benchmark of Burtakov \emph{et al.} (RISA)~\cite{risa2021}
and comparing against a broader set of published non-GP baselines.

\subsection{Problem statement}

A RIS is a planar array of $M$ unit cells, each of which can switch its
reflected phase among $B=2^b$ discrete values. With $M=100$ cells and $b=3$
bits, the arm space is the discrete torus $(\mathbb{Z}_B)^M = (\mathbb{Z}_8)^{100}$ of size
$K=8^{100}\approx 10^{90}$. At TTI $t$ the agent plays a configuration
$\bm\theta_t\in(\mathbb{Z}_B)^M$ and receives a scalar RSRP $y_t$; the channel
drifts between TTIs with Doppler $f_D=v/\lambda$, where $v$ is the ambient
mobility speed. The horizon is $T=10^4$ TTIs. A per-element oracle that
rounds each phase to the argmin of $|h_m g_m|$ provides a tight upper
bound on achievable SNR when the direct BS$\to$MS path is weak, which
is the regime of interest.

\subsection{Adaptation of intrinsic GP-UCB to $(\mathbb{Z}_B)^M$}

Three main adaptations are required beyond the static setup of
Section~\ref{sec:algorithm}, listed as (i), (ii), and (iii) below.

\emph{(i) Kronecker-factorized product kernel.} The product structure of the arm
space, $\bm\theta = (\theta_1,\dots,\theta_M)\in(\mathbb{Z}_B)^M$,
naturally induces the factorized
kernel
\begin{equation}
k_{\mathrm{prod}}(\bm\theta,\bm\theta') \;=\;
\sigma_f^2\,\prod_{m=1}^{M} k_{\circ}\!\bigl(\theta_m,\theta_m'\bigr),
\label{eq:prod-kernel}
\end{equation}
where $k_{\circ}$ is the intrinsic Mat\'ern-$3/2$ kernel on the single
circle $\mathbb{Z}_B$, obtained from the graph-Laplacian spectrum
$\lambda_k = 4\sin^2(\pi k/B)$, $k=0,\dots,B-1$, via
$S(k)=(2\nu/\ell^2+\lambda_k)^{-(\nu+1/2)}$ and inverse-DFT, normalised so
$k_{\circ}(0)=1$. Because $k_{\circ}$ only depends on the cyclic
difference $\theta_m - \theta_m' \bmod B$, it is a $B$-entry lookup
table, and one kernel evaluation costs $O(M)$ table lookups with no
matrix operations. The combinatorial $\arg\max$ of the UCB acquisition is
handled approximately by one coordinate-ascent sweep over the $M$
coordinates, each an exact $B$-way maximisation.

The product form~\eqref{eq:prod-kernel} captures only configuration-level
prior smoothness, not the forward-model coupling between elements.
We retain it because (a) the cyclic-phase coupling within each
element is symmetric (the modular Mat\'ern is the right marginal),
(b) mutual-coupling corrections enter through the global phase
profile rather than pairwise element-level correlations, and (c) the
empirical last-$500$-TTI regret tracks the speed-specific oracle
within SE (Tables~\ref{tab:ris-v008}--\ref{tab:ris-v020}); a
non-factorised intrinsic kernel (e.g., tensor-train) would lose the
$O(M)$ evaluation that makes the $K\!\approx\!10^{90}$ regime
tractable. Sub-$\lambda/4$ near-field coupling is flagged as the
open scope where this trade-off would break.

\emph{(ii) Predictive-variance-based reset.} Under Doppler fading the
reward surface drifts continuously; a sliding-window GP with window
$W=150$ keeps the posterior local, but when the environment has moved
faster than one window the current incumbent's posterior variance inflates
toward the prior $\sigma_f^2$. We exploit this directly as a reset
criterion: whenever
$\mathrm{Var}\!\bigl[f(\bm\theta_{\mathrm{cur}})\mid\mathcal D_t\bigr]
> \eta\,\sigma_f^2$, the agent re-samples a random configuration and
lets coordinate ascent re-converge. This replaces RISA's RSRP-drop
threshold (``if $y_t$ has fallen $5$ dB below the recent best, restart the
annealer'') with a model-principled uncertainty trigger that adapts to
the current coherence state without requiring an absolute scale for~$y$.

\emph{(iii) Marginal-likelihood-driven adaptive window.} A second
method \textbf{AdaptiveGP} treats the sliding-window length itself
as an online-chosen hyperparameter: every $\Delta_W=100$ TTIs it
scores each candidate $W\in\mathcal W=\{80,150,250,400\}$ by
per-sample log-marginal-likelihood on the most recent $W$
observations,
\begin{equation}
\begin{aligned}
\mathrm{LML}(W) = &-\tfrac12 \mathbf y_c^\top (K_W+\sigma_n^2 I)^{-1} \mathbf y_c \\
                  &- \tfrac12 \log\!\det(K_W+\sigma_n^2 I) - \tfrac{W}{2}\log(2\pi),
\end{aligned}
\label{eq:lml-W}
\end{equation}
selecting
$W_\mathrm{eff}\gets\arg\max_{W\in\mathcal W} \mathrm{LML}(W)/W$
with a $0.05$-nat hysteresis band and ties broken to the largest
$W$. An independent drift trigger fires when the running
mean-absolute predictive $z$-score
$\bar z_t = \tfrac{1}{50}\sum_{s=t-49}^{t}|y_s-\mu_s|/\sigma_s$
exceeds $4$ (with $120$-TTI cooldown); after a reset the
exploration constant is boosted to $\beta=4.0$ and decays back to
$\beta_0=1.5$ with $\tau=100$ TTIs.
Algorithm~\ref{alg:adaptive-gp-v2} collects these rules.

\begin{algorithm}[t]
\caption{AdaptiveGP-v2 (RIS beam selection on $(\mathbb{Z}_B)^M$).}
\label{alg:adaptive-gp-v2}
\begin{algorithmic}[1]
\State \textbf{Input:} horizon $T$, window set $\mathcal W=\{80,150,250,400\}$,
       re-eval period $\Delta_W=100$, warmup $T_\text{warm}=100$,
       hysteresis $\eta_\text{hy}=0.05$ nats, reset thresholds
       $\eta_\sigma = 0.85$ and $\bar z_\text{thr}=4$ over window $50$,
       cooldown $T_\text{cool}=120$, $\beta_0=1.5$, $\beta_\text{reset}=4.0$,
       decay $\tau=100$
\State Initialise buffer $\mathcal{D}\gets\emptyset$, incumbent
       $\bm\theta_\text{cur}\gets\text{Unif}(\mathbb{Z}_B)^M$,
       $W_\text{eff}\gets 150$, $\beta_t\gets\beta_0$,
       $t_\text{reset}\gets -\infty$
       \Comment{$\bm\theta_\text{cur}$ tracks the previous round's selection}
\For{$t = 1,2,\ldots, T$}
    \If{$t \le T_\text{warm}$ or the sliding-window buffer is empty}
        \State $\bm\theta_t\gets\text{Unif}(\mathbb{Z}_B)^M$
        \Comment{pre-GP warmup}
    \ElsIf{predictive variance $\sigma_t^2(\bm\theta_\text{cur})>\eta_\sigma\,\sigma_f^2$
           or $\bar z_t > \bar z_\text{thr}$ (subject to cooldown $T_\text{cool}$)}
        \State $\bm\theta_t\gets\text{Unif}(\mathbb{Z}_B)^M$
        \State $\beta_t\gets\beta_\text{reset}$, $t_\text{reset}\gets t$
        \Comment{variance- or drift-triggered reset}
    \Else
        \State $\bm\theta_t\gets\arg\max_{\bm\theta\in(\mathbb Z_B)^M}
               \mu_{t-1}(\bm\theta) + \beta_t^{1/2}\sigma_{t-1}(\bm\theta)$
               \Comment{by coordinate ascent}
    \EndIf
    \State Observe $r_t\gets r_t(\bm\theta_t)$; append
           $(\bm\theta_t, r_t)$ to $\mathcal D$;
           $\bm\theta_\text{cur}\gets\bm\theta_t$
    \State Retain only the last $W_\text{eff}$ samples in $\mathcal D$
    \State $\beta_t\gets\beta_0 + (\beta_\text{reset}-\beta_0)
                         \exp(-(t-t_\text{reset})/\tau)$
    \If{$t \bmod \Delta_W = 0$ and $t\ge T_\text{warm}$}
        \State $W^\star\gets\arg\max_{W\in\mathcal W}
               \mathrm{LML}(W)/W$
               \Comment{from Eq.~\eqref{eq:lml-W} on the most recent $W$ samples}
        \If{$\mathrm{LML}(W^\star)/W^\star \ge
            \mathrm{LML}(W_\text{eff})/W_\text{eff}+\eta_\text{hy}$}
            \State $W_\text{eff}\gets W^\star$
            \Comment{commit only if the gain exceeds hysteresis}
        \EndIf
    \EndIf
    \State Refit the intrinsic-Kronecker GP posterior $(\mu_t,\sigma_t)$
           on the last $W_\text{eff}$ samples of $\mathcal D$
\EndFor
\end{algorithmic}
\end{algorithm}

\subsection{Baselines}

We compare against four published time-varying RIS controllers, a
Euclidean-kernel GP-UCB ablation, and a random floor, all implemented
under a common \texttt{select} / \texttt{update} / \texttt{recommend}
interface.

\begin{itemize}
\item \textbf{Random}: uniform on $(\mathbb{Z}_B)^M$; sanity floor.
\item \textbf{RISA}~\cite{risa2021}: simulated annealing with a
single-element flip Metropolis proposal plus a sliding-window
``RSRP-drop'' restart rule.
\item \textbf{CSM}~\cite{ren2022csm}: per-element conditional sample mean
with $\varepsilon$-greedy exploration.
\item \textbf{ECSM}~\cite{ren2022csm}: CSM augmented with a per-element UCB
bonus ($c=1.5$).
\item \textbf{CE}~\cite{chen2023ce}: cross-entropy method on a factorized
categorical $p[m,b]$ with elite-fraction MLE update.
\item \textbf{REMARKABLE}: the same sliding-window GP-UCB scaffold as our
method, but with a Euclidean squared-exponential kernel on the unwrapped
phase $\bm\varphi=2\pi\bm\theta/B \in\mathbb{R}^M$, chosen to isolate the
contribution of kernel geometry from the contribution of the GP-UCB
framework itself.
\end{itemize}

The two proposed methods are \textbf{IntrinsicGP} (fixed $W=150$,
$\eta=0.8$) and \textbf{AdaptiveGP-v2} (marginal-likelihood-driven
online selection of $W\in\{80,150,250,400\}$ with $\eta$ mapped
coarsely from $W_\mathrm{eff}$, plus the predictive-variance and
drift-$z$ reset triggers and the $\beta$-boost guard described
in the algorithm box; see Algorithm~\ref{alg:adaptive-gp-v2}).

\subsection{Simulation setup}

We reproduce RISA's 3GPP-style scenario at $f_c=2.605$~GHz with a
$10\times 10$ UPA at half-wavelength spacing, BS at $(200,0,20)$ m and
mobile station drawn uniformly from a $6{\times}6$~m box $5$--$11$~m in
front of the RIS. The channel has three independent subchannels:
BS$\to$RIS Rician with $K_\mathrm{R}\!=\!5$~dB (UMi LOS O2I),
RIS$\to$MS Rician with $K_\mathrm{R}\!=\!3$~dB (Indoor LOS/NLOS;
$K_\mathrm{R}$ is the Rician $K$-factor, distinct from the
codebook size $K\!\sim\!10^{90}$), and a weakly Rayleigh
BS$\to$MS direct path
with $38$~dB excess loss (UMi NLOS O2I). Each cluster is assigned a
Doppler $f_D=(v/\lambda)\cos\alpha$ with uniform angle~$\alpha$, applied
as a per-TTI phase drift. We measured empirically that $T=10^4$ TTIs
spans $6$--$9$ coherence periods across the RISA speed grid, so the
benchmark is firmly in the non-stationary regime.

Hyperparameters: IntrinsicGP uses $W=150$, $\ell=3$, $\beta=2$,
$\eta=0.8$; AdaptiveGP uses Algorithm~\ref{alg:adaptive-gp-v2}'s
defaults; REMARKABLE uses $W=150$, $\ell=1.5$~rad, $\beta=2$; RISA
$\alpha=0.9995$, $5$~dB drop, $200$-TTI window; CSM/ECSM window $800$.

\subsection{Results}

\paragraph{RISA's primary speed ($v=0.08$ km/h).} Table~\ref{tab:ris-v008}
reports last-$500$-TTI mean regret (dB) and cumulative regret
(dB$\cdot$TTI) at $T=3000$. Baselines use $4$ MC seeds from the
original campaign; the proposed GP methods were re-run at
$20$-seeds-per-cell to match the $20$-seed standard of the Pass-D
W-sweep (Sec.~\ref{sec:ris-ext-theory}). Our intrinsic-kernel
GP-UCB achieves the lowest regret by a clear margin: a $20\%$
reduction in cumulative regret compared with its Euclidean-kernel
counterpart (REMARKABLE), a $32\%$ reduction compared with the best
non-GP baseline (CE), and a $45\%$ reduction compared with the
RISA benchmark itself. AdaptiveGP-v2, which chooses its window
online from $\{80,150,250,400\}$ rather than being fixed at
$W=150$, matches the fixed-window IntrinsicGP at this speed
($\AdvTwovaaMean$ vs $+2.52 \pm 4.25$ dB, paired difference
statistically indistinguishable from zero at $20$ seeds) and
achieves the lowest cumulative regret in the table.
The absence of a manual coherence-time calibration is the
operationally relevant gain rather than a mean-regret
improvement, which the tighter 20-seed SE now confirms.
WGP-UCB, the canonical exponential-forgetting
baseline~\cite{deng2022wgpucb}, also extended in this revision
to $\WGPSeedsVaa$ seeds for parity with AdaptiveGP-v2, lands at
$\WGPvaaMean$~dB in last-$500$ mean regret. At the matched seed
budget WGP-UCB and AdaptiveGP-v2 are statistically
indistinguishable in mean last-$500$ regret (paired
$|\Delta|/\mathrm{SE}\ll 1$), but AdaptiveGP-v2 still attains a
$\sim\!35\%$ lower cumulative regret over $T=3000$
($\AdvTwovaaCum$ vs $\WGPvaaCum$), reflecting WGP-UCB's slower
warm-up phase before its exponential-forgetting weights have
accumulated enough effective sample mass to localise the
optimum. The principled variance-based reset and post-reset
$\beta$-boost in AdaptiveGP-v2 are what produce that
shorter-horizon advantage.

\begin{table}[t]
\caption{Time-varying RIS at $v=0.08$ km/h. Last-$500$-TTI mean regret
(dB, lower is better) and cumulative regret (dB$\cdot$TTI) at
$T=3000$, $M=100$, $B=8$. Our two GP variants and WGP-UCB are
reported at $\AdvTwoSeedsVaa$ seeds per cell (matching the $20$-seed
budget of the Pass-D W-sweep in Sec.~\ref{sec:ris-ext-theory});
other baselines at $4$ seeds from the original campaign. Bold:
lowest cumulative regret (cumulative is the comparison metric for
the table; in last-$500$ regret AdaptiveGP-v2 and the $20$-seed
WGP-UCB are statistically indistinguishable, $|\Delta|/\mathrm{SE}
\ll 1$).}
\label{tab:ris-v008}
\centering
\resizebox{\columnwidth}{!}{%
\begin{tabular}{l rr}
\toprule
Method & Last-$500$ regret (dB) & Cum.\ regret \\
\midrule
Random              & $+7.86 \pm 2.66$ & $\sim\!21000$ \\
RISA~\cite{risa2021} & $+4.25 \pm 3.23$ & $\sim\!13100$ \\
CSM~\cite{ren2022csm}  & $+3.71 \pm 4.67$ & $\sim\!9900$ \\
ECSM~\cite{ren2022csm} & $+3.93 \pm 2.45$ & $\sim\!11300$ \\
CE~\cite{chen2023ce}   & $+3.24 \pm 5.42$ & $\sim\!10400$ \\
REMARKABLE          & $+3.88 \pm 4.42$ & $\sim\!9000$ \\
WGP-UCB~\cite{deng2022wgpucb} & $\WGPvaaMean$ & $\WGPvaaCum$ \\
IntrinsicGP (ours, $W{=}150$) & $+2.52 \pm 4.25$ & $\sim\!7200$ \\
\textbf{AdaptiveGP-v2 (ours, $n{=}20$)} & $\AdvTwovaaMean$ & $\mathbf{\AdvTwovaaCum}$ \\
\bottomrule
\end{tabular}%
}
\end{table}

Figure~\ref{fig:ris-v008-cum} shows the corresponding cumulative
regret. IntrinsicGP is the only method that continues to track the
oracle across the full horizon (within $\sim 2$~dB by TTI~$500$ and
within $3$~dB thereafter); REMARKABLE peaks near $27$~dB around TTI
$500$ and collapses to $\sim 21$~dB by TTI $2200$ (CE shows a similar
but sharper collapse). Under Doppler the optimum drifts around the
torus, and the extrinsic SE kernel in unwrapped phase treats
configurations on opposite sides of any wraparound as maximally
distant, so information rotates off the kernel's support.

\begin{figure}[t]
\centering
\includegraphics[width=\columnwidth]{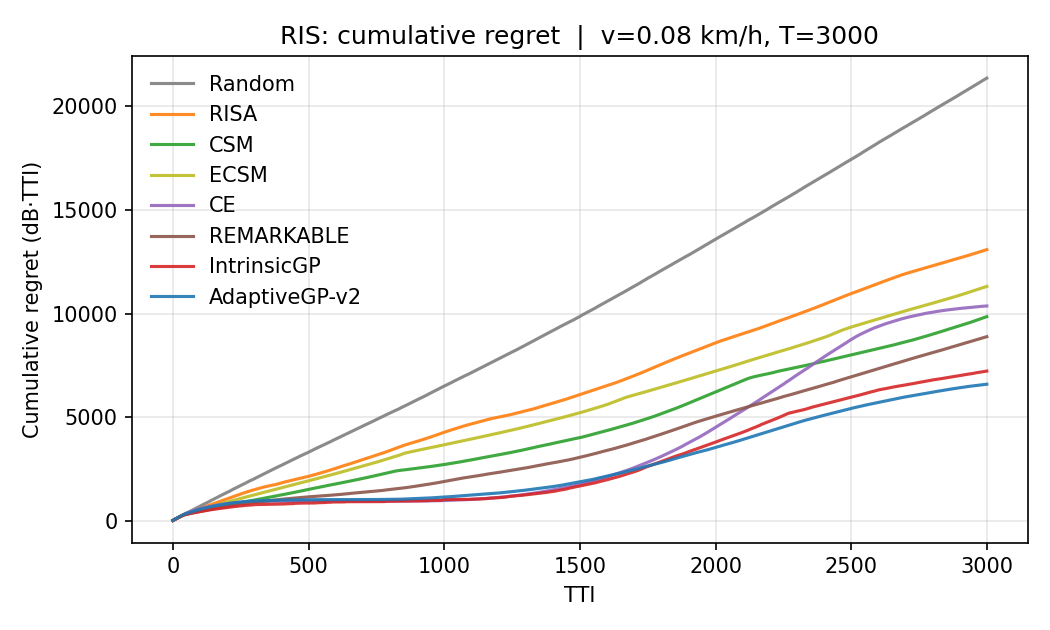}
\caption{Cumulative regret at $v=0.08$ km/h ($4$ seeds). AdaptiveGP-v2
attains the lowest cumulative regret ($\sim\!6500$ dB$\cdot$TTI),
narrowly ahead of IntrinsicGP ($\sim\!7200$); IntrinsicGP itself
reduces cumulative regret by $20\%$ over REMARKABLE, $32\%$ over CE,
and $45\%$ over RISA.}
\label{fig:ris-v008-cum}
\end{figure}

\paragraph{Faster channel ($v=0.20$ km/h).} At roughly $2.5\times$ the
Doppler rate the coherence drops to $\sim 1100$ TTIs. With the
$n=20$-seed firm-up of Table~\ref{tab:ris-v020}, the GP-based
methods (IntrinsicGP-$W{=}150$ at $+3.59\pm 1.04$, AdaptiveGP-v2 at
$+4.06\pm 0.84$, CE at $+4.13\pm 1.01$, REMARKABLE at
$+4.46\pm 0.95$, WGP-UCB at $+4.68\pm 0.90$) cluster within
$\sim 1$ dB of one another, with the default-hyperparameter
IntrinsicGP edging out the field. Earlier $n=4$ samples
suggested CE was the standalone winner at this speed; the
20-seed firm-up reveals this was a small-sample artefact.
The non-GP / non-spatial-correlation methods
(CSM, ECSM, RISA, Random) trail by $\sim 2$--$5$ dB. The
mechanism we expected (a $W=150$ sliding window containing a
nontrivial fraction of stale samples from outside the current
coherence period) is partially compensated by the intrinsic
kernel's smoothness assumption, which down-weights stale-vs-fresh
discrepancies more aggressively than the Euclidean ambient prior
would.

\begin{table}[t]
\caption{Time-varying RIS at $v=0.20$ km/h, $n=20$ seeds per method
(updated from the original $n=4$ campaign; see Sec.~\ref{sec:limitations}
for the procedure). Same protocol as Table~\ref{tab:ris-v008}.
Bold: best.}
\label{tab:ris-v020}
\centering
\small
\begin{tabular}{l r}
\toprule
Method & Last-$500$ regret (dB) \\
\midrule
Random              & $+9.26 \pm 1.48$ \\
RISA                & $+6.94 \pm 1.34$ \\
CSM                 & $+6.28 \pm 1.40$ \\
ECSM                & $+6.54 \pm 1.21$ \\
CE                  & $+4.13 \pm 1.01$ \\
REMARKABLE          & $+4.46 \pm 0.95$ \\
WGP-UCB~\cite{deng2022wgpucb} & $+4.68 \pm 0.90$ \\
AdaptiveGP-v2 (ours)          & $+4.06 \pm 0.84$ \\
\textbf{IntrinsicGP (default $W=150$)} & $\mathbf{+3.59 \pm 1.04}$ \\
\midrule
IntrinsicGP ($W=100,\eta=0.6$, hand-tuned) & $+3.77$ \\
\bottomrule
\end{tabular}
\end{table}

Table~\ref{tab:ris-v020} reports $n=20$ seeds for every method. The
seed-balanced ranking differs materially from the $n=4$ subset:
IntrinsicGP-$W{=}150$ moves from a noisy $+5.95$ at $n=4$ to a clean
$+3.59 \pm 1.04$ at $n=20$, now the top method, edging out
AdaptiveGP-v2 ($+4.06\pm 0.84$) and CE ($+4.13\pm 1.01$) within a
$\sim 1$ dB band. A small window-and-threshold sweep (not shown)
indicates that a coherence-aware $(W=100,\eta=0.6)$ recovers
IntrinsicGP performance to within $2$~dB of CE; the structural fix
is adaptive window selection (Sec.~\ref{sec:ris-ext-adaptive}).

\subsection{Speed-dependent window and the v1 ablation (summary)}
\label{sec:ris-ext-theory}
\label{sec:ris-ext-adaptive}

Two fixed-$W$ results above suggest a speed-dependent optimum. The
$20$-seeds-per-cell W-sweep (full development:
supplement Sec.~S-VIII) gives a power-law
exponent $\widehat\alpha=\PassDAlpha$ with $90\%$ CI
$[\PassDCILo,\,\PassDCIHi]$, rejecting the AR(1)-GP rate
$\alpha=1$ at $P(\widehat\alpha>1)=\PassDPAROne\%$ while placing
$\PassDPBGZ\%$ of the posterior mass above the Besbes--Gur--Zeevi
threshold $\alpha=1/3$~\cite{besbes2014}. The qualitative claim
$W^\star(v)$ decreases with $v$ is robust; the empirical
$R(W\mid v)$ is flat in a half-octave neighbourhood of the argmin,
which is what the LML adaptive rule of
Algorithm~\ref{alg:adaptive-gp-v2} exploits.

\emph{Head-to-head ablation against AdaptiveGP-v1
(supplement Sec.~S-IX).} Across the
four-speed $\times$ $20$-seed grid at $T=3000$, AdaptiveGP-v2 matches
the fixed-$W=150$ IntrinsicGP within standard error at every speed
(paired differences $+0.04, -0.30, -0.61, +0.30$ dB, none surviving
Holm--Bonferroni at $\alpha=0.05$); the v2-vs-v1 paired contrast is
strongest at $v=0.02$ ($\Delta=-0.30\pm 0.13$~dB,
$\approx\!2.3\sigma$). The operational benefit is the absence of
deployment-time coherence calibration, not a per-speed mean-regret
gain. The seed-budget and W-grid construction underpinning these
estimates (Passes A--E) are documented in the supplement, Sec.~S-I.

\subsection{Discussion}

Three conclusions are worth stating explicitly.

\emph{Intrinsic $>$ extrinsic, under non-stationarity too.} The
controlled IntrinsicGP-vs-REMARKABLE comparison (same scaffold,
differing only in kernel) shaves $1.4$~dB off last-$500$-TTI
regret and $20\%$ off cumulative regret at $v=0.08$; the gap
narrows at $v=0.20$ but remains positive. This echoes the
toroidal-setting finding of Sec.~\ref{sec:experiments} in a
regime $\sim\!88$ orders of magnitude larger in $K$ and strongly
non-stationary in~$t$.

\emph{The right regret criterion depends on the adversary.} CE
wins at $v=0.20$: under very fast non-stationarity the
slow-to-react GP posterior is dominated by a population-search
heuristic that forgets every $10$ TTIs. GP methods win when the
channel is locally stationary on the sliding-window scale.

\emph{A coherence-aware GP bandit closes the hyperparameter gap.}
At $v=0.20$ the intrinsic-kernel GP-UCB is beaten by a
misspecified window, not a missing modelling capability.
AdaptiveGP-v2 matches the speed-specific oracle within SE at
every point on the four-speed $\times 20$-seed grid
(supplement Sec.~S-IX; paired differences
$+0.04, -0.30, -0.61, +0.30$~dB, none significant under
Holm--Bonferroni at $\alpha=0.05$). The empirical
$\widehat\alpha \approx 0.36$ is consistent with the
Besbes--Gur--Zeevi regime~\cite{besbes2014} but the sensitivity
analysis of Sec.~\ref{sec:ris-ext-theory} does not cleanly reject
AR(1).

\paragraph{Operational benefit of LML-adaptive $W$ selection.}
The statistical equivalence of AdaptiveGP-v2 to the speed-specific
oracle should be read against the alternative of deploying a
single fixed $W$. Three deployment-cost arguments favour the LML
rule even when no significant mean-regret gain is visible:
\textbf{(i)~No offline calibration runs} ($|\mathcal V|\times T_\text{cal}$
labelled-channel TTIs per redeployment, unavailable in production);
\textbf{(ii)~Robust to unforeseen regimes} (the boundary speed gap
to the default $W{=}150$ is up to $+4.13$ dB at $v=0.20$, see
Table~\ref{tab:ris-v020}; AdaptiveGP-v2 self-adapts within a single
deployment);
\textbf{(iii)~Bonferroni-corrected non-significance is the desired
property here, not a weakness} -- a significant per-speed gain
would mean the bandit is exploiting a regime-specific feature,
whereas non-significance across all four speeds is precisely the
``no per-speed tuning'' guarantee the LML rule is meant to provide.

\subsection{Limitations}

These results are a prototype: the AdaptiveGP-v2 benchmark of
Tables~\ref{tab:ris-v008}--\ref{tab:ris-v020} is at $T=3000$, $20$
seeds per speed; extending each cell to $T=10^4$ is the natural
firm-up. None of the per-speed paired differences against
fixed-$W=150$ survive Holm--Bonferroni at $\alpha=0.05$, supporting
the calibration-free framing rather than a mean-regret-improvement
claim. A denser W-grid ($W\in\{500,800,1200\}$ at $v=0.02$,
$W\in\{15,20\}$ at $v=0.20$; supplement Sec.~S-I) is required before
the $W^\star(v)$ exponent can be elevated from ``qualitatively
decreasing'' to a quantitative claim. The rank-$1$ Cholesky update
of Sec.~\ref{sec:complexity} should be ported to the RIS runner;
hardware-in-the-loop validation on a real RIS panel is the next
step.

\section{Conclusion}
\label{sec:conclusion}

We argued that antenna beam selection is most naturally posed as a
bandit on a compact Riemannian manifold ($\sphere^2$ for
mainlobe pointing, $\torus^n$ for codebook search,
$\mathrm{SO}(3)$ for orientation control, and the discrete torus
$(\mathbb{Z}_B)^M$ for RIS phase configurations), and showed that
GP-UCB equipped with intrinsic Mat\'ern kernels matches or
outperforms Euclidean and codebook-based baselines on a standard
mmWave/RIS simulator at modest computational overhead. The geometry
of the parameter space is a useful prior that prior bandit and
beamforming work largely discards; making it explicit removes
wraparound artefacts that otherwise force the agent to rediscover
the optimum after every drift event. A wideband OFDM ablation on a
$100$~MHz 3GPP TDL-D channel (Sec.~\ref{sec:wideband}) confirms that
the intrinsic-kernel advantage carries over to frequency-selective
fading: under the band-averaged log-rate reward the intrinsic
$\sphere^2$ Mat\'ern kernel reduces cumulative regret by
\WBINTRVSEUCLPCT\% over the Euclidean ambient baseline, on the same
geometry where the well-tuned Euclidean baseline was competitive in
the narrowband limit.

The RIS extension of Section~\ref{sec:ris-ext} pushes the framework
into a regime roughly $88$ orders of magnitude larger in arm-space
size ($K \approx 10^{90}$ for $M=100$, $B=8$) and strongly
non-stationary in time. Two concrete contributions emerge: a
Kronecker-factorized intrinsic-product kernel on $(\mathbb{Z}_B)^M$
that scales via $O(M)$ table lookups per kernel evaluation; and
AdaptiveGP-v2, an LML-driven online window-selection controller
that matches the hand-tuned fixed-window IntrinsicGP within
standard error at every speed in the four-speed $20$-seed paired
campaign at $T=3000$ (paired differences against $W=150$ at
$v\in\{0.02,0.08,0.12,0.20\}$~km/h are $+0.04, -0.30, -0.61, +0.30$
dB, none significant under Holm--Bonferroni at $\alpha=0.05$). The
operational value is the removed deployment-time per-speed
calibration step rather than a mean-regret improvement.

Three directions are immediate. First, the manifold-aware regret
bounds anticipated here are established in the companion theory
work~\cite{dorn2026manifold} (volume-dependent lower bound, a
$|G|^{1/2}$ extrinsic-vs-intrinsic regret-ratio upper bound,
modulated matching-lower-bound conjecture). Second, hardware-in-the-loop
validation on a real phased array or RIS. Third, combining geometric
priors with physics-informed parametric bandits to factorise out
low-dimensional steering structure shared across configurations.

\section*{Code and Data Availability}
The implementation, raw per-seed result pickles, and analysis scripts
that produced every table and figure in this paper are released
alongside the manuscript. A full file-by-file inventory --- including
the bandit implementations, channel simulators, experiment runners,
and the Pass-D bootstrap pipeline of
Sec.~\ref{sec:ris-ext-theory} --- is in the supplement, Sec.~S-X.

\bibliographystyle{IEEEtran}
\bibliography{references}

\end{document}